\documentclass[a4paper,11pt]{article}
\usepackage[english]{babel}
\usepackage[T1]{fontenc} %256 bit font encoding

\usepackage[table,usenames,dvipsnames]{xcolor}
\usepackage{tikz}
\usetikzlibrary{
  positioning,
  fadings,
  }

\usepackage{authblk}

\usepackage{amssymb} %lots of math symbols
\usepackage{mathtools} %for, e.g., \DeclareMathOperator
\usepackage{amsmath} %for spacing in equations and aligns
\usepackage{amsfonts}
\usepackage{dsfont}
\usepackage{mathrsfs} % for D
\usepackage{paralist}
\usepackage{enumitem}
\usepackage{hyperref}
\hypersetup{colorlinks=true}
\usepackage{amsthm}
\usepackage[margin=3cm]{geometry}
\usepackage[numbers,sort&compress]{natbib} 
\usepackage{graphicx}
\usepackage[font=small,labelfont=bf,indention = \parindent]{caption}%to format the figure captions
\usepackage{complexity}

\newtheorem{theorem}{Theorem}
\newtheorem{definition}[theorem]{Definition}
\newtheorem{lemma}[theorem]{Lemma} %counter name in square brackets
\newtheorem{proposition}[theorem]{Proposition}
\newtheorem{corollary}[theorem]{Corollary}
\newtheorem{observation}[theorem]{Observation}
\newtheorem{implication}[theorem]{Implication}

\newcommand{\e}{\mathrm{e}}
\renewcommand{\i}{\mathrm{i}}
\DeclareMathOperator{\Tr}{Tr}

\DeclareMathOperator{\ran}{ran}
\renewcommand{\d}{\mathrm{d}}
\DeclareMathOperator{\rank}{rank}

\DeclareMathOperator{\Herm}{Herm}
\DeclareMathOperator{\Pos}{Pos}
\DeclareMathOperator{\CPT}{CPT}
\DeclareMathOperator*{\argmin}{arg\,min}

\renewcommand{\vec}{\operatorname{vec}} %vectorization
\newcommand{\U}{\operatorname{U}} %unitary group
\renewcommand{\O}{\operatorname{O}} %orthogonal group
\newcommand{\cone}{\operatorname{cone}} %conic hull
\newcommand{\sign}{\mathrm{sign}} 

\newcommand{\CC}{\mathbb{C}}
\newcommand{\RR}{\mathbb{R}}

\newcommand{\1}{\mathds{1}}
\renewcommand{\EE}{\mathbb{E}}

\DeclareMathAlphabet{\mathbbold}{U}{bbold}{m}{n}
\newcommand{\mc}[1]{\mathcal{#1}}
\renewcommand{\A}{\mc{A}}
\newcommand{\V}{\mc{V}}
\renewcommand{\W}{\mc{W}}
\newcommand{\norm}[1]{\left\Vert #1 \right\Vert}%norm
\newcommand{\normn}[1]{\lVert #1 \rVert}
\newcommand{\normb}[1]{\bigl\Vert #1 \bigr\Vert}

\newcommand{\dnorm}[1]{\norm{#1}_\diamond} %diamond norm

\newcommand{\snorm}[1]{\norm{#1}} %spectral norm

\newcommand{\snormb}[1]{\normb{#1}}

\newcommand{\nnorm}[1]{\norm{#1}_\ast} %nuclear norm

\newcommand{\nnormb}[1]{\normb{#1}_\ast}
\newcommand{\fnorm}[1]{\norm{#1}_\fro} %nuclear norm
\newcommand{\fnormn}[1]{\normn{#1}_\fro}
\newcommand{\fnormb}[1]{\normb{#1}_\fro}
\newcommand{\jnorm}[1]{\norm{#1}_\mysquare} % \jnorm{J(M)}=dim(V)*\dnorm{M}

\newcommand{\kw}[1]{\frac{1}{#1}}
\newcommand{\tkw}[1]{\tfrac{1}{#1}}
\newcommand{\argdot}{{\,\cdot\,}}

\renewcommand{\L}{\operatorname{\mathrm{L}}}
\renewcommand{\M}{\operatorname{\mathbb{L}}}
\newcommand{\DM}{\operatorname{\mc{D}}}

\newcommand{\st}{\mathrm{subject\ to}}

\newcommand{\fro}{\mathrm{F}}
\newcommand{\minimize}{\mathrm{minimize }}
\newcommand{\maximize}{\mathrm{maximize }}

\providecommand{\DC}{\operatorname{\mathscr{D}}} %descent cone
\newcommand{\ad}{^\ast}%symbol for conjugate transpose

\newcommand{\ev}{\epsilon} %error vector
\newcommand{\AG}{\A_G}
\newcommand{\Astr}{\A_{\mathrm{str}}}
\newcommand{\capstr}[1]{\textbf{#1}}%to structure the captions
\newcommand{\eps}{\mathrm{eps}}
\newcommand{\mysquare}{{\protect\scalebox{0.5}{$\square$}}}
\newcommand{\fd}{f_{\mysquare}}

\makeatletter % metadata for hyperref
 \hypersetup{pdftitle = {Improving compressed sensing with the diamond norm},
	     pdfauthor = {M. Kliesch, R. Kueng, J. Eisert, D. Gross},
	     pdfsubject = {Compressed sensing, inverse problems, state estimation, system recovery, quantum mechanics},
	     pdfkeywords = {Compressive sensing, 
			    quantum process tomography, 
			    tensor recovery, 
			    blind deconvolution, 
			    convex relaxation, 
			    bilinear compressed sensing,
			    low-rank matrix recovery,
			    bowling scheme}
	    }
\makeatother

\begin{document}
% \title{Diamond norm as regularizer for low-rank matrix recovery}
\title{Improving compressed sensing \\ with the diamond norm
}

% %		\thanks{
% 		This work was presented in part at 
% 				conference on Quantum Information Processing and Communication (QIPC) in Leeds, UK, 2015;
% 				on the workshop on Applied Harmonic Analysis and Sparse Approximation in Oberwolfach, Germany, 2015; 
% 				at the Matheon Conference on Compressed Sensing and its Applications in Berlin, Germany, 2015; 
% 				and the 2016 Information Theory Workshop.
% 		}%

\author[1,2]{M.\ Kliesch}
\author[2]{R.\ Kueng}
\author[1]{J.\ Eisert}
\author[2]{D.\ Gross}

\affil[1]{\small Dahlem Center for Complex Quantum Systems, Freie Universit\"{a}t Berlin, Germany}

\affil[2]{\small Institute for Theoretical Physics, University of Cologne, Germany}

\date{}

\maketitle

% \IEEEpeerreviewmaketitle
% submission in cs.IT, quant-ph, math.OC

%%% -------------------------------------------------------------------------
%%% --------------------------- Abstract ------------------------------------
%%% -------------------------------------------------------------------------
\begin{abstract}
In low-rank matrix recovery, one aims to reconstruct a low-rank matrix from a minimal number of linear measurements. Within the paradigm of compressed sensing, this is made computationally efficient by minimizing the nuclear norm as a convex surrogate for rank. 

In this work, we identify an improved regularizer based on the so-called diamond norm, a concept imported from quantum information theory. We show that --for a class of matrices saturating a certain norm inequality-- the descent cone of the diamond norm is contained in that of the nuclear norm. This suggests superior reconstruction properties for these matrices. We explicitly characterize this set of matrices. Moreover, we demonstrate numerically that the diamond norm indeed outperforms the nuclear norm in a number of relevant applications: These include signal analysis tasks such as blind matrix deconvolution or the retrieval of certain unitary basis changes, as well as the quantum information problem of process tomography with random measurements.

The diamond norm is defined for matrices that can be interpreted as order-4 tensors and it turns out that the above condition depends crucially on that tensorial structure. In this sense, this work touches on an aspect of the notoriously difficult tensor completion problem.
\end{abstract}

%%% ================================================
%%% =================== Intro  =========================
%%% ================================================
\section{Introduction}
The task of recovering an unknown low-rank matrix from a small number
of measurements appears in a variety of contexts. Examples of this
task are provided by collaborative
filtering in machine learning \cite{KorBellVol09},
quantum state tomography in quantum information
\cite{GroLiuFla10,FlaGroLiu12},  the estimation of covariance
matrices \cite{BicLev08,CheChiGol15}, or face recognition
\cite{BasJac03}. If the measurements are
linear, the technical problem reduces to identifying the lowest-rank
element in an affine space of matrices. In general, this problem is
$\NP$-hard and it is thus unclear how to approach it algorithmically
\cite{Nat95}. 

In the wider field of compressed sensing \cite{FouRau13}, the strategy
for treating such problems is to replace the complexity measure -- 
here the rank -- with a tight convex relaxation. Often, it can be
rigorously proved that the resulting convex optimization problem has
the same solution as the original problem, 
while at the same time allowing for an efficient algorithm.
The tightest (in some sense \cite{FazHinBoy01}) convex relaxation of
rank is the \emph{nuclear  norm}, i.e.\ the sum of singular values.
Minimizing the nuclear norm subject to linear constraints is a
semi-definite program and a great number of rigorous performance
guarantees has been provided for low-rank reconstruction using
nuclear norm minimization
\cite{CanRec09,CanTao10,RecFazPar10,GroLiuFla10,Gro11,CanPla11,AhmRecRom12,KueRauTer15,KabKueRau15}.

The geometry of convex reconstruction schemes is now well-understood
(c.f.\ Figure~\ref{fig:geometry}). Starting with a convex regularizer
$f$ (e.g.\ the nuclear norm), geometric proof techniques like Tropp's
Bowling scheme \cite{Tro14} or Mendelson's small ball method
\cite{Men14,KolMen14} bound the reconstruction error in
terms of the descent cone of $f$ at the matrix that is to be
recovered. Moreover, these arguments suggest that the error would
decrease if another convex regularizer with smaller descent cone would
be used. This motivates the search for new convex regularizers that (i) are
efficiently computable and (ii) have a smaller descent cone at particular points of interest.

In this work, we introduce such an improved regularizer based on the
\emph{diamond norm} \cite{KitSheVya02}. 
This norm plays a fundamental role in the 
context of quantum information and operator theory \cite{Pau02}. 
For this work, it is convenient to also use a variant of the diamond norm that we call the \emph{square norm}. 
While not obvious from its definition, it has been found that the
diamond norm can be efficiently computed by means of a semidefinite
program (SDP) \cite{Wat09,BenTaS09,Wat12}.  
Starting from one such SDP characterization \cite{Wat12}, we identify
the set of matrices for which the square norm's descent cone is contained
in the corresponding one of the nuclear norm. 
As a result, low-rank matrix recovery guarantees that have been
established via analyzing the nuclear norm's descent cone
\cite{Tro14,KueRauTer15} are also valid for
square norm regularization, provided that the matrix of interest
belongs to said set. 
What is more,
bearing in mind the reduced size of the square norm's
descent cone, we actually expect an improved recovery. Indeed, with numerical
studies we show an improved performance. 

Going beyond low-rank matrix recovery, we identify several applications. 
In physics, we present numerical experiments that show that the
diamond norm offers improved performance for \emph{quantum process
tomography} \cite{ShaKosMoh11}. The goal of this important task is to
reconstruct a quantum process from suitable preparations of inputs and
measurements on outputs extending quantum \emph{state} tomography,
for which low-rank methods have been studied extensively
\cite{GroLiuFla10,FlaGroLiu12,Vor13,Kue15}. 
We then identify applications to problems from the context of signal
processing.  These include matrix versions of the
\emph{phase retrieval problem}
\cite{Wal63,CanEldStr11,CanLi13,CanStroVor13,AleBanFicMix14,CanLiSol14,GroKraKue15_partial,GroKraKue15_masked}, 
as well as a matrix version of the \emph{blind deconvolution problem}
\cite{AhmRecRom12}.
Recently, a number of \emph{bi-linear problems} combined with sparsity or low-rank
structures have been investigated in the context of compressed sensing, with first progress on recovery
guarantees being reported \cite{AhmRecRom12,WalJun13}.
The present work can be seen as a contribution to this recent development.

We conclude the introduction on a more speculative note.
The diamond norm is defined for linear maps
taking operators to operators -- i.e., for objects that can also be
viewed as order-$4$ tensors. 
We derive a characterization of those
maps for which the diamond norm offers improved recovery, and find
that it depends on the order-$4$ tensorial structure. In this sense,
the present work touches on an aspect of the notoriously difficult
\emph{tensor recovery problem} (no canonic approach or reference seems
to have emerged yet, but see Ref.\ \cite{RauSto15} for an up-to-date list of
partial results).
In fact, the ``tensorial nature'' of the diamond norm was the original
motivation for the authors to consider it in more detail as a
regularizer -- even though the eventual concrete applications we found
do not seem to have a connection to tensor recovery. It would be
interesting to explore this aspect in more detail.

%%% =====================================================
%%% ================== Preliminaries  =====================
%%% =================================================
\section{Preliminaries}
In this section, we introduce notation and mathematical preliminaries
used to state our main results. 
We start by clarifying some notational conventions. 
In particular, we introduce certain matrix norms and the partial trace for operators acting on a tensor product space. Moreover, we summarize a general geometric setting for the convex recovery of structured signals. 

\subsection{Vectors and operators}
\label{sec:notation}
For a positive integer $n$ we use the notation $[n] \coloneqq \{1,2, \dots, n\}$. 
Throughout this work we focus exclusively on finite dimensional and mostly complex vector spaces $\V, \W$ whose elements we mostly denote by lower case Latin letters, e.g. $x \in \V$. 
One can also set $\V = \CC^n$ and $\W = \CC^N$ throughout the paper. 
However, as low-rank matrix completion is basis independent and in order to avoid ambiguity, we will still refer to them as $\V$ and $\W$. 

We assume that each vector space $\V$ is equipped with an inner product $\langle \cdot, \cdot \rangle_\V$ -- or simply $\langle \cdot, \cdot \rangle$ for short -- that is linear in the second argument. 
Such an inner product induces the Euclidean norm
\begin{equation}
\fnorm{x} \coloneqq \sqrt{ \langle x, x \rangle_\V} \quad \forall x \in \V
\end{equation}
and moreover defines a conjugate linear bijection from $\V$ to its dual space $\V^\ast$:
to any $x \in \V$ we associate a dual vector $x\ad \in \V^\ast$ which is uniquely defined via $x\ad y = \langle x, y \rangle_\V$ $\forall y \in \V$. 
The vector space of linear maps from $\V$ to $\W$ is denoted by $\L(\V \to \W)$. Its elements being \emph{operators} are denoted by capital Latin letters (e.g. $X,Y,U,V$) and often we also refer to them as matrices. 
Indeed, for $\V = \CC^n$ and $\W = \CC^N$ an operator $X \in \L(\V \to \W)$ is given by a complex $N \times n$ matrix.
We also write $\L(\V) = \L (\V \to \V)$ for the sake of notational brevity. 
The \emph{adjoint} $X\ad \in \L ( \W \to \V)$ of an operator $X \in \L(\V \to \W)$ is determined by 
$\langle X\ad x,y \rangle_\V = \langle x, X y \rangle_\W$ for all $x \in \V$ and $y \in \W$
If $X$ is given by a matrix, then $X\ad$ is given by the complex conjugated and transposed matrix. 
We call an operator $X \in \L(V)$ \emph{self-adjoint}, or Hermitian, if $X\ad = X$. 
A self-adjoint operator $X$ is \emph{positive semidefinite}, if it has a non-negative spectrum. 
A particularly simple example for such an operator is the identity operator $\1_{\V} \in \L(\V)$. 
The set of positive semidefinite operators in $\L(\V)$ forms a convex cone which we denote by $\Pos(\V)$ \cite{Bar02}. 
This cone induces a partial ordering on $\L (\V)$ and we write $X \succeq Y$ if $X-Y \in \Pos(\V)$. 
On $\L(\V)$ we define the \emph{Frobenius} (or Hilbert-Schmidt) \emph{inner product} to be
\begin{equation}\label{eq:inner_prod}
\langle X,Y \rangle_{\L(\V)} \coloneqq \Tr( X\ad Y) \quad \forall X,Y \in \L(\V), 
\end{equation}
where $\Tr(Z)$ denotes the trace of an operator $Z \in \L(\V)$. 
By $\rank(X)$ we denote the \emph{rank}, i.e., the number of non-zero singular values of $X \in \L(\V)$. 
In addition to that, we are going to require three different matrix norms
\begin{alignat}{2}
\nnorm{X} &\coloneqq 
\Tr\bigl( \sqrt{ X\ad X}\, \bigr) &\quad & \text{(nuclear norm/trace norm)}, 
\\
\fnorm{X} &\coloneqq 
\sqrt{ \langle X,X \rangle} &&\text{(Frobenius norm)},
\\
\snorm{X} &\coloneqq 
\sup_{x \in \V} \frac{\fnorm{ X x }}{\fnorm{x}} && \text{(spectral norm)}.
\label{eq:snrom_def}
\end{alignat}
The Frobenius norm is induced by the inner product
\eqref{eq:inner_prod}, while the nuclear norm requires the operator square root: for $X \in \Pos (\V)$ we let $\sqrt{X} \in \Pos(\V)$ be the unique positive semi-definite operator obeying $\sqrt{X}^2 = X$. Note that these norms coincide with the Schatten $1$-, Schatten $2$- and Schatten $\infty$-norms, respectively. 
All Schatten norms are multiplicative under taking tensor products.
The Frobenius norm is preserved under 
any regrouping of indices, the prime example of such an operation
being the vectorization of matrices. This fact justifies our convention to 
extend the notation $\fnorm{\argdot}$ to the $2$-norms of vectors and (later on)
tensors.

A crucial role is played by the space of \emph{bipartite operators}
$\L(\W\otimes \V)$, 
by which we refer to operators that act on a tensor product space. 
For such operators we define the \emph{partial trace}
$\Tr_\W : \L(\W\otimes \V) \to \L(\V)$ as the linear extensions of the map given by
\begin{equation}\label{eq:TrW}
 \Tr_\W(B\otimes A) \coloneqq \Tr(B) \, A \, ,
\end{equation}
where $A \in \L(\V)$ and $B \in \L(\W)$, see also Figure~\ref{fig:tensor_diagram_TrW}.
When the underlying vector spaces are again written as $\V = \CC^n$ and $\W = \CC^N$, a bipartite operator $X \in \L(\W \otimes \V)$ is given by an array $X=(x_{i,j,k,l})_{j,l \in [n],\, i,k \in [N]}$.
Then $\Tr_\W(X)$ is given by an $n\times n$ matrix with components $\sum_{i=1}^N x_{i,j,i,l}$. 

Finally, we define our improved regularizer on $\L(\W\otimes \V)$ to be
\begin{equation}\label{eq:jnorm_variational}
 % \begin{split}
 \jnorm{X} \coloneqq
 \max\bigl\{%&
 \nnorm{(\1_\W \otimes A)X(\1_\W \otimes B)}: \
 	 % \\&
	    A,B \in \L(\V), \
 	 % \\&	  
	\fnorm{A}=\fnorm{B} = \sqrt{\dim(\V)} \bigr\}
 \, .
  % \end{split}
\end{equation}
It is easy to see that $\jnorm{\argdot}$ is a norm and we call it the \emph{square norm}. It will become clear later on that the square norm is closely related to the diamond norm $\dnorm \argdot$ from quantum information theory \cite{Wat09}. 
More explicitly, as we will discuss in Section~\ref{sec:notation_maps}, 
$\jnorm{X} = \dim(\V)\, \dnorm{J^{-1}(X)}$, where $J$ denotes the so-called Choi-Jamio{\l}kowski isomorphism. Both, square and diamond norm can be calculated by a semidefinite program (SDP) satisfying strong duality \cite{Wat12}. 
Also, note that the pair $A = B = \1_\V$ is admissible in the maximization \eqref{eq:jnorm_variational}. Inserting it recovers $\nnorm{X}$ and establishes the bound $\nnorm{X} \leq \jnorm{X}$. This bound plays a crucial role for our results. 

\newcommand{\Mbox}[1]{
  \node (M) [Bbox] at (0,0) {#1};
  \def\d{.35}
  \def\len{.5}
  \path (M.west) ++ (0,\d\baselineskip) coordinate (Moli);
  \path (M.west) ++ (0,-\d\baselineskip) coordinate (Muli);
  \path (M.east) ++ (0,\d\baselineskip) coordinate (More);
  \path (M.east) ++ (0,-\d\baselineskip) coordinate (Mure);
  }
\begin{figure}
\centering
%%% ---------------- tikzpicture ------------------
\leavevmode
\beginpgfgraphicnamed{fig1}%
\definecolor{niceblue}{rgb}{0.33,0.5,0.8}%
\tikzset{%
   sbox/.style = {draw, rounded corners = .5ex,%
		   minimum height = 1.5\baselineskip,%
		   minimum width = 1.8em},%standard box
   blau/.style = {top color=niceblue!12,%
		   bottom color=niceblue!90},%blue shading
   Bbox/.style = {sbox, blau},% blue box
   leg/.style = {rounded corners = .5ex,thick},%tensor legs
   dir/.style = {gray,thick}%direction
   }%
\begin{tikzpicture}[node distance = 1ex]%
 \node (M2){
   \begin{tikzpicture}
   \Mbox{$X$}
   \draw [leg] (More) -- ++(\len,0) --++(0,\len) node (Xo) [near end, right] {$\V$};
   \draw [leg] (Mure) -- ++(\len,0) --++(0,-\len) node (Xu) [near end, right] {$\V^\ast$};
   \draw [leg] (Moli) -- ++(-\len,0) --++(0,\len) node (Yo) [near end, left] {$\W$};
   \draw [leg] (Muli) -- ++(-\len,0) --++(0,-\len) node (Yu) [near end, left] {$\W^\ast$};
   \draw [dir, <-] (Xo.north east) ++(1ex,0) coordinate (ore) -- (ore|-Xu.south);
   \end{tikzpicture}
   };
 \node (TrY) [right = of M2]{\Large{$\overset{\Tr_\W}{\mapsto}$}};
 \node (M3) [right = of TrY]{
   \begin{tikzpicture}
   \Mbox{$X$}
   \draw [leg] (Moli) -- ++(-\len,0) |- (Muli);
   \draw [leg] (More) -- ++(\len,0) --++(0,\len) node (Xo) [near end, right] {$\V$};
   \draw [leg] (Mure) -- ++(\len,0) --++(0,-\len) node (Xu) [near end, right] {$\V^\ast$};
   \draw [dir, <-] (Xo.north east) ++(1ex,0) coordinate (ore) -- (ore|-Xu.south);
   \end{tikzpicture}
 };
\end{tikzpicture}
\endpgfgraphicnamed
%%% ------------------------------------
%%% compile with:
%%% pdflatex --jobname=fig1 diamonds.tex
	% \includegraphics{fig1}% include pre-compiled file
%%% ------------------------------------
  \caption{Tensor network diagrams: tensors are denoted by boxes with one line for each index. Contraction of two indices corresponds to connection of the corresponding lines. 
	     \newline
    \capstr{Left:} A bipartite operator 
      $X \in \L(\W\otimes \V)$ viewed as a tensor in $\W \otimes \V \otimes \W^\ast \otimes \V^\ast$, i.e., as a tensor with four indices. \newline
    \capstr{Right:} Its partial trace $\Tr_\W(X)$ as an operator on $\V$.}
  \label{fig:tensor_diagram_TrW}
\end{figure}
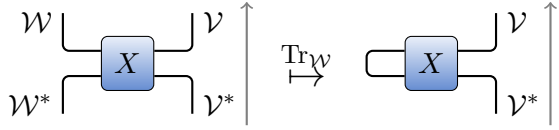

\subsection{Convex recovery of structured signals}
\label{sub:convex_recovery}
In this section, we summarize a recent but already widely used geometric proof technique for low-rank matrix recovery. 
Mainly following the exposition of Ref.\ \cite{Tro14}, we devote this section to explaining the general reconstruction idea.

In the setting of convex recovery of structured signals, one obtains a \emph{measurement vector} $y \in \CC^m$ of a \emph{signal} $x_0 \in \V$ in some vector space $\V$ via a \emph{measurement map} $\A : \V \to \CC^m$, 
\begin{equation}\label{eq:measurement}
 y = \A(x_0)+\ev \, ,
\end{equation}
where $\ev\in \CC^m$ represents additive noise in the sampling process. Throughout, we assume linear data acquisition, i.e., that $\A$ is linear. 

The goal is to efficiently obtain a good approximation to $x_0$ given $\A$ and $y$ for the case where one only has knowledge about some structure of $x_0$. Of course, it is desirable that the number $m$ of measurements $y_i$ required for a successful reconstruction is as small as possible. For several different structures of the signal $x_0$ a general approach of the following form has proven to be very successful \cite{ChaRecPar12}. One chooses a convex function $f : \V \to \overline \RR$ that reflects the structure of $x_0$ and performs the following convex minimization
\begin{equation}\label{eq:general_reconstruction}
 x^f_\eta = \argmin\{ f(x) : \ \fnorm{\A(x) - y} \leq \eta\}\, ,
\end{equation}
where $\eta \geq 0$ is some anticipated error bound.

Next, we give two definitions and a general error bound that has proven to be helpful to find such recovery guarantees. The \emph{descent cone} of a convex function is the set of non-increasing directions $u$. From the convexity of the function, it follows that the descent cone is a convex cone. The following definitions can also be found, e.g., in Ref.\ \cite{Tro14}.

\begin{definition}[Descent cone]\label{def:DC}
	The \emph{descent cone} $\DC(f,x)$ of a proper convex function $f: \V \to \overline \RR$ at the point $x\in \V$ is 
	\begin{equation}
	\DC(f,x) \coloneqq \bigcup_{\tau>0} \{u\in \V : \ f(x+\tau u) \leq f(x) \} \, .
	\end{equation}
\end{definition}

The \emph{minimum singular value} of a linear map $\A$ is the minimal value of $\fnorm{\A(x)}$ taken over all $x$ with $\fnorm x = 1$. Restricting this minimization to a cone yields the \emph{minimum conic singular value}. 

\begin{definition}[Minimum conic singular value] 
Let $\A:\V \to \CC^m$ be a linear map and $K \subset \V$ be a cone. 
The minimum singular value of $\A$ with respect to the cone $K$ is defined as 
\begin{equation}
 \lambda_{\min{}}(\A; K) \coloneqq \inf_{x \in K} \frac{\fnorm{\A(x)}}{\fnorm{x}} \, .
\end{equation}
\end{definition}

The following proposition is the basis for many recovery guarantees. 
In terms of the tangent cone of the unit ball of $f$ it has been proved in Ref.~\cite{ChaRecPar12} and was later restated in terms of the descent cone by Tropp. 

\begin{proposition}[Error bound for convex recovery, Tropp's version \cite{Tro14}]\label{prop:general_reconstruction}
 Let $x_0 \in \V$ be a signal, $\A\in \L(\V \to \CC^m)$ be a measurement map, $y = \A(x_0)+\ev$ a vector of $m$ measurements with additive error $\ev \in \CC^m$, and $x^f_\eta$ be the solution of the optimization \eqref{eq:general_reconstruction}. If $\fnorm{\ev} \leq \eta$ then
 \begin{equation}
  \fnormb{x^f_\eta - x_0}\leq \frac{2\eta}{\lambda_{\min{}} (\A;\DC(f,x_0))} \, .
 \end{equation}
\end{proposition}

Note that the statement in Ref.\ \cite{Tro14} shows this result for real vector spaces
only. However, taking a closer look at the proof reveals that it also holds for complex vector spaces as well.
We make the following simple but important observation: 
\begin{observation}[Improved recovery]\label{obs:smaller_DC}
	The smaller the descent cone the better the recovery guarantee.
\end{observation}

An important example is low-rank matrix recovery. 
Here, $x_0 = X_0$ is some $n\times N$ matrix with $\rank(X_0)=r$. A low rank $r$ provides the structure that allows for a reconstruction from significantly fewer measurements than the  dimension \mbox{$n\cdot N$} of the ambient space. 
For this case, choosing $f = \nnorm{\argdot}$ to be the nuclear norm has proven very successful, as the nuclear norm is the convex envelope of the matrix rank \cite{FazHinBoy01}. 
In order to give a concrete bound, consider a real matrix $X_0$ of rank $r$ and $m$ measurements $y_j = \Tr\bigl(A_j^\dagger X_0\bigr)+\ev_j$ with each $A_j$ being a real random matrix with entries drawn independently from a normalized Gaussian distribution. 
Then one can show that (see, e.g., Ref.\ \cite{Tro14})
\begin{equation}\label{eq:lambda_nnorm_Gaussian}
% \begin{split}
 \lambda_{\min{}}(\A;\DC(\nnorm{\argdot},X_0)) 
%  \\&
 \geq \sqrt{m-1} -\sqrt{3r(n_1+n_2 - r)} - t
%  \end{split}
\end{equation}
with probability $1-\e^{-t^2/2}$ (over the random measurements). 
As a consequence, a number of $\gtrsim 3 \rank(X_0)(n_1+n_2-\rank(X_0))$ measurements are enough for a successful reconstruction of the real-valued matrix $X_0$ with high probability. 

%%% ===============================================
%%% =================== Results =======================
%%% ===============================================
\section{Results}
We show that for certain structured recovery problems, replacing the regularizer $f$ in a convex recovery \eqref{eq:general_reconstruction} by an optimized regularizer $\fd$ can potentially improve performance; see also Figure~\ref{fig:geometry}. 
For the case where $f$ is the nuclear norm and $\fd$ the square norm, we show such an improvement with numerical simulations in Section~\ref{sec:application_maps}. 

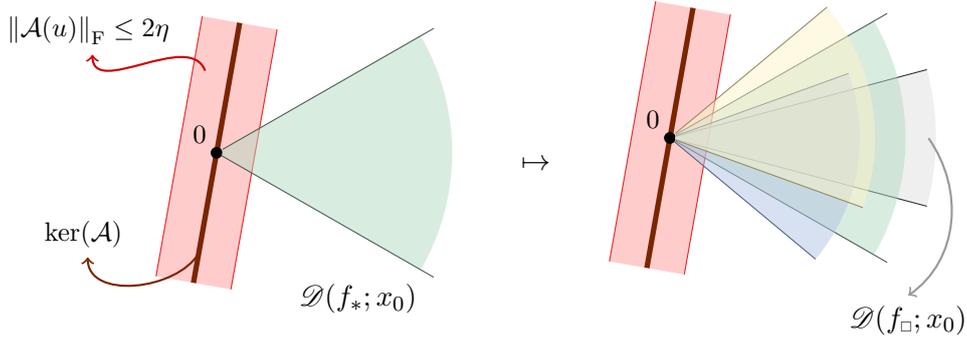
\begin{figure*}
\centering
%%% ---------------- tikzpicture ------------------
\leavevmode
\beginpgfgraphicnamed{fig2}%
\def\len{3.5}%
\def\breite{.5}%
\def\rot{-10}%
\def\rotc{-10}%
\colorlet{Acol}{NavyBlue}%
\colorlet{Bcol}{Goldenrod}%
\colorlet{Ccol}{ForestGreen}%
\definecolor{Dcol}{gray}{0.6}%
	%def cones
	\def\coneAlen{2.5}%
	\def\coneAl{-40-\rotc}%
	\def\coneAr{20-\rotc}%
	\def\coneBlen{2.7}%
	\def\coneBl{-20-\rotc}%
	\def\coneBr{40-\rotc}%
	\def\coneClen{3.1}%
	\def\coneCl{-30-\rotc}%
	\def\coneCr{30-\rotc}%
	\def\coneDlen{3.5}%
	\def\coneDl{-15-\rotc}%
	\def\coneDr{15-\rotc}%
\begin{tikzpicture}[
ker/.style = {line width = 2pt, Brown},
bker/.style = {Red},
schlauch/.style = {red!40,semitransparent},
coneA/.style = {Acol!30,semitransparent},
bconeA/.style = {Acol!50!black},
coneB/.style = {Bcol!30,semitransparent},
bconeB/.style = {Bcol!50!black},
coneC/.style = {Ccol!30,semitransparent},
bconeC/.style = {Ccol!30!black},
coneD/.style = {Dcol!30,semitransparent},
bconeD/.style = {Dcol!30!black},
lab/.style = {transform shape = false}, %labels
]
\node (left) {
	\begin{tikzpicture}[rotate=\rot,transform shape]
	%set 4 corner coordinates of the schlauch
	\path (0,0) coordinate (ker-li) ++ (\breite,0) coordinate (ker-li-o) -- ++(0,\len) coordinate (ker-re-o);
	\path (ker-li) ++ (-\breite,0) coordinate (ker-li-u) -- ++(0,\len) coordinate (ker-re-u);
	% fill schlauch and draw bounding lines
	\fill [schlauch] (ker-li-u) rectangle (ker-re-o);
	\draw [bker] (ker-li-o) -- (ker-re-o);
	\draw [bker] (ker-li-u) -- (ker-re-u);
	\draw [ker]  (ker-li)  -- ++(0,\len) coordinate (ker-re) node [inner sep = 1pt] (origin) [midway]{};
	% origin
	\filldraw (origin) circle (2pt) 
	node[above left, lab]{\small $0$};
	%def cones
	%cone C (green)
	\path (origin) -- ++(\coneCr:\coneClen) coordinate (cCre) ++(\coneCr:.2) coordinate (bcCre);
	\path (origin) -- ++(\coneCl:\coneClen) coordinate (cCli) ++(\coneCl:.2) coordinate (bcCli);
	\fill [coneC] (origin) -- (cCli) arc (\coneCl:\coneCr:\coneClen) -- (origin);
	\draw [bconeC](origin) -- (bcCli)
	(origin) -- (bcCre);
	%labels
	% kerA
	\path (ker-li-u) ++ (-.4,.5) 
	node [lab,anchor = east] (kerA) {\small $\ker(\A)$} ;
	\draw [->,Brown,thick] (ker-li)++(0,.4) to [out=-110, in = -80] (kerA);
	% \A(\Phi)
	\path (origin) -- (ker-re-u) 
	node (schlauch) [midway]{};
	\path (schlauch) -- ++ (-.5,.6) node [lab,anchor = east] (eta) 
	{\small $\fnorm{\A(u)} \leq 2\eta$};
	\draw [->, red!80!black, thick]  (schlauch) to[out=120, in=-80] (eta);
	% cone C
	\path (cCli) ++(0,1) coordinate (anchorC)
	++(-.8,-1.2) node (DCC) [lab,anchor = north]{$\DC(f_{\ast};x_0)$};
	%   \draw [->, Ccol, thick] (anchorC) to[out=-60, in =60] (DCC.north);
	\end{tikzpicture}
}; %node (left)
\path (left.east) ++(.5,0) node (imp) [anchor = west] {$\mapsto$};
\path (imp.east) ++(.5,-.1) node (right) [anchor = west]{
	\begin{tikzpicture}[rotate=\rot, transform shape]
	%set 4 corner coordinates of the schlauch
	\path (0,0) coordinate (ker-li) ++ (\breite,0) coordinate (ker-li-o) -- ++(0,\len) coordinate (ker-re-o);
	\path (ker-li) ++ (-\breite,0) coordinate (ker-li-u) -- ++(0,\len) coordinate (ker-re-u);
	% fill schlauch and draw bounding lines
	\fill [schlauch] (ker-li-u) rectangle (ker-re-o);
	\draw [bker] (ker-li-o) -- (ker-re-o);
	\draw [bker] (ker-li-u) -- (ker-re-u);
	\draw [ker]  (ker-li)  -- ++(0,\len) coordinate (ker-re) node [inner sep = 1pt, anchor = center] (origin) [midway]{};
	% origin
	\filldraw (origin) circle (2pt) 
	node[above left, lab]{\small $0$};
	%cone D (gray)
	\path (origin) -- ++(\coneDr:\coneDlen) coordinate (cDre);
	\path (origin) -- ++(\coneDl:\coneDlen) coordinate (cDli);
	\fill [coneD] (origin) -- (cDli) arc (\coneDl:\coneDr:\coneDlen) -- (origin);
	\draw [bconeD](origin) -- (cDli) 
	(origin) -- (cDre);
	%cone C (green)
	\path (origin) -- ++(\coneCr:\coneClen) coordinate (cCre) ++(\coneCr:.2) coordinate (bcCre);
	\path (origin) -- ++(\coneCl:\coneClen) coordinate (cCli) ++(\coneCl:.2) coordinate (bcCli);
	\fill [coneC] (origin) -- (cCli) arc (\coneCl:\coneCr:\coneClen) -- (origin);
	\draw [bconeC](origin) -- (bcCli)
	(origin) -- (bcCre);
	%cone A (blue)
	\path (origin) -- ++(\coneAr:\coneAlen) coordinate (cAre);
	\path (origin) -- ++(\coneAl:\coneAlen) coordinate (cAli);
	\fill [coneA] (origin) -- (cAli) arc (\coneAl:\coneAr:\coneAlen) -- (origin);
	\draw [bconeA](origin) -- (cAli) 
	(origin) -- (cAre);
	%cone B (yellow)
	\path (origin) -- ++(\coneBr:\coneBlen) coordinate (cBre);
	\path (origin) -- ++(\coneBl:\coneBlen) coordinate (cBli);
	\fill [coneB] (origin) -- (cBli) arc (\coneBl:\coneBr:\coneBlen) -- (origin);
	\draw [bconeB](origin) -- (cBli) 
	(origin) -- (cBre);
	%labels
	% kerA
	\path (origin) -- (ker-re-u) 
	node (schlauch) [midway]{};
	\path (cDre) -- (cDli) coordinate [midway] (Dmw);
	\path (Dmw)++(.02,0) coordinate(anchorD) ++(.1,-2.1) 
	node (DCD) [lab,anchor = north]{$\DC(\fd;x_0)$};
	\draw [->, Dcol, thick] (anchorD) to[out=-45, in =45] (DCD.north);
	\end{tikzpicture}
}; %node (left)
\end{tikzpicture}
\endpgfgraphicnamed
%%% ------------------------------------
	% \includegraphics{fig2}% include pre-compiled file
%%% ------------------------------------
\caption{Extension of the geometric arguments \cite{Tro14} used to establish Proposition~\ref{prop:general_reconstruction}. The descent cone $\DC(\fd;x_0)$ of the optimized regularizer $\fd$ is contained in an intersection of descent cones.} 
\label{fig:geometry}
\end{figure*}

\begin{proposition}[Optimizing descent cones]
  \label{prop:construction}
	Let $C \subset \V$ be a convex set and $I$ be a compact index set. 
	Moreover, let $\{f_i\}_{i \in I}$ be a family of upper semi-continuous convex functions $f_i:C \to \RR$. Define another convex function $\fd$ as the point-wise supremum $\fd(x) \coloneqq \sup_{i \in I} f_i(x)$. Then
	\begin{equation}
	\DC(\fd;x) 
	\subset \bigcap_{i \in I(x)} \DC(f_i;x)
	\end{equation}
	for any $x\in C$, where $I(x) \coloneqq \{i \in I : f_i(x) = \fd(x)\}$ is the active index set at $x$, where we use the convention 
	$\bigcap_{i \in \emptyset}\DC(f_i;x) \coloneqq \V$. 
\end{proposition}

\begin{proof}[Proof of Proposition \ref{prop:construction}]
By $\cone(S) \coloneqq \bigcup_{\tau>0} \{ \tau s: s\in S\}$ we will denote the cone generated by a set $S$.
According to Definition \ref{def:DC} of the descent cone, we have 
\begin{equation}
  \DC(\fd; x) 
  = 
  \bigcup_{\tau>0} \{ u \mid \sup_{i\in I} f_i(x+\tau u) \leq \fd(x)\} \, .
\end{equation}
Writing the supremum as an intersection yields
\begin{align}
 \DC(\fd; x) 
 &= 
 \bigcup_{\tau>0} \bigcap_{i\in I} \{\tau u \mid f_i(x+u) \leq \fd(x)\}
 \\
  &\subset
  \bigcap_{i\in I} \cone \{u \mid f_i(x+u) \leq \fd(x)\}\, .
  \label{eq:cap_cone_set}
\end{align}
By $B_\epsilon \subset \V$ we denote the ball around the origin of radius $\epsilon$. Now, consider a non-active index $i \in I\setminus I(x)$. As $f_i$ is upper semi-continuous, there exists $\epsilon>0$ such that for all $u \in B_\epsilon$ we have $f_i(x+u) < \fd(x)$. Hence, $B_\epsilon \subset \{u \mid f_i(x+u) \leq \fd(x)\}$, so that the corresponding cone in Eq.~\eqref{eq:cap_cone_set} is the entire space. Therefore, every non-active index $i$ can be omitted in the intersection, 
\begin{equation}
 \DC(\fd; x) 
 \subset 
 \bigcap_{i\in I(x)} \cone \{u \mid f_i(x+u) \leq f_i(x)\}\, .
\end{equation}
The definition of the descent cone of $f_{i}$ finishes the proof.
\end{proof}

The square norm \eqref{eq:jnorm_variational} is a particular instance of such a supremum over nuclear norms. Thanks to the following nuclear norm bound \eqref{eq:nnorm_leq_jnorm}, Proposition~\ref{prop:construction} can lead to an improved recovery for any bipartite operator $X \in \L(\W\otimes \V)$ satisfying 
\begin{equation}\label{eq:equality}
  \nnorm{X} = \jnorm{X} \, . 
\end{equation}
Here, we will only need the lower bound on the square norm but, in
order to fully relate it to the usual matrix norms, we also provide
two upper bounds.

\begin{proposition}[Bounds to the square norm]\label{prop:bounds_on_diamond_norm}
	For any $X \in \L(\W\otimes \V)$ 
	\begin{align}
	\nnorm{X} &\leq \jnorm{X} \, ,  \label{eq:nnorm_leq_jnorm}  
	\\
	\jnorm{X} & \leq \dim (\V)\, \nnorm{X} \, , \label{eq:jnorm_leq_nnorm}
	\\
	\jnorm{X} & \leq \dim(\W\otimes \V) \, \snorm{X}  \, . \label{eq:jnorm_leq_snorm}
	\end{align} 
\end{proposition}
A proof of this proposition is given in Section~\ref{sec:SDPs}.

Our second main result fully characterizes the set of operators satisfying Eq.~\eqref{eq:equality}, i.e., saturating inequality~\eqref{eq:nnorm_leq_jnorm}. 
As we will see below, for such operators, recovery guarantees for square norm reconstructions can be inherited from those of the nuclear norm. 

\begin{theorem}[Extremal operators] \label{thm:extremality}
	Let $X \in \L(\W\otimes \V)$ be a bipartite operator. 
	Then Eq.~\eqref{eq:equality} holds if and only if 
	\begin{equation}\label{eq:structure}
	\Tr_{\W} \bigl( \sqrt{ XX\ad} \bigr)
	= \Tr_{\W} \bigl( \sqrt{ X\ad X} \bigr)\,=
	\frac{\nnorm{X}}{\dim(\V)} \,  \1_{\V}.
	\end{equation}
\end{theorem}

For now, we content ourselves with sketching the proof idea and present the full proof later. 

\begin{proof}[Proof idea]
For the case where Eq.~\eqref{eq:equality} is satisfied, we single out a primal feasible optimal point. Exact knowledge of this point together with complementary slackness then allows us to severely restrict the range of possible dual optimal points. Relation \eqref{eq:structure} is an immediate consequence of these restrictions.

To show the converse, we insert a particular feasible point into the dual SDP of the square norm.
Eq.~\eqref{eq:structure} enables us to explicitly evaluate the objective function at this point. Doing so yields $\nnorm{X}$ which, in turn, implies $\jnorm{X} \leq \nnorm{X}$ by weak duality.
Combining this implication with the converse bound from Proposition~\ref{prop:bounds_on_diamond_norm}
establishes $\snorm{X}=\jnorm{X}$, as claimed.
\end{proof}

As an implication of Theorem~\ref{thm:extremality} and Proposition~\ref{prop:construction} we obtain the following. 

\begin{corollary}[Intersection of descent cones]\label{cor:subset}
Let $X \in \L(\W\otimes \V)$ satisfy Eq.~\eqref{eq:equality}. Then
  \begin{equation}
%   \begin{split}
    \DC(\jnorm \argdot; X) 
%     \\&
      \subset 
    \bigcap_{(A,B)\in I(X)} \DC(\nnorm{(\1_\W \otimes A)(\argdot)(\1_\W \otimes B)}; X) \, ,
%     \end{split}
  \end{equation}
  where $I(X)$ contains all $A,B \in \L(V)$ with 
  $\fnorm{A}=\fnorm{B}=\sqrt{\dim(\V)}$ 
  and being active in the sense that 
$\jnorm{X} = \nnorm{(\1_\W \otimes A) X(\1_\W \otimes B)}$. 
\end{corollary}

Setting $A = B = \1_\V$ gives an element of $I(X)$  and yields the inclusion
\begin{equation}\label{eq:inclusionDC}
 \DC(\jnorm \argdot; X) \subset \DC(\nnorm{\argdot};X)
\end{equation} for any $X$ satisfying Eq.~\eqref{eq:equality}. 
As an immediate application, we will see in the next section that the square norm inherits recovery guarantees from the nuclear norm for signals $X$ satisfying $\jnorm{X} = \nnorm{X}$. 
In the case where $\jnorm{X} \neq \nnorm{X}$, the inclusion of descent cones \eqref{eq:inclusionDC} does, in general, not hold. Indeed, we have observed in numerical experiments that the usual nuclear norm reconstruction performs better in that case. 

%%% ==================================================
%%% ================= Applications =========================
%%% ==================================================
\section{Applications to low-rank matrix recovery}
\label{sec:application_low_rank}
In this section we focus on low-rank matrix recovery of Hermitian bipartite operators $X_0 \in \L (\W \otimes \V)$ satisfying the condition~\eqref{eq:equality}
that are either real-valued or complex-valued.
As already mentioned in Section~\ref{sub:convex_recovery}, there the task is to efficiently recover an unknown matrix $X_0$ of low-rank at most $r$ 
from $m$ noisy linear measurements of the form
\begin{equation}
y_i = \Tr \left(A_i X_0 \right) + \epsilon_i, \quad i = 1,\ldots, m \, , 
\end{equation}
where $A_1,\ldots,A_m \in \L (\W \otimes \V)$ are the measurement matrices and $\epsilon_1,\ldots, \epsilon_m \in \RR^m$ denotes additive noise in the sampling process. 
By introducing a measurement map $\A: \L (\W \otimes \V) \to \RR^m$ of the form $\A (X_0) = \sum_{i=1}^m \Tr \left(A_i X_0 \right)e_i$, where $e_1,\ldots,e_m$ denotes the standard basis in $\RR^m$, the entire measurement process can be summarized as
\begin{equation} \label{eq:measurement_operator}
y = \A (X_0) + \epsilon \, .
\end{equation}
Here, $y = \left( y_1,\ldots,y_m \right)^T\in \RR^m$ contains all measurement outcomes and $\epsilon \in \RR^m$ denotes the noise vector. 
If a bound $\fnorm{\epsilon} \leq \eta$ on the noise is available,
many measurement scenarios have been identified
%\cite{CanRec09,CanTao10,ChaRecPar12,RecFazPar10,Gro11,CanPla11,ForRauWar11,Liu11,Tro14,KueRauTer15,KabKueRau15}, 
where estimating $X_0$ by
\begin{equation}\label{eq:Mast}
  X^\ast_\eta \coloneqq \argmin\{ \nnorm{X}: \fnorm{\A(X)-y} \leq \eta \} 
\end{equation}
stably recovers $X_0$.
Note that by employing the well-known SDP formulation of the nuclear norm \cite{Boy04} this optimization can be recast as
\begin{equation}
\arraycolsep=1.5pt
\begin{array}{r l l}
X^\ast_\eta
=&\underset{X,Y,Z}{\argmin} & \frac{1}{2}\bigl( \Tr(Y) + \Tr(Z) \bigr)
\\[.3em]
&\st &
	\begin{pmatrix}
	Y & -X \\ -X\ad & Z
	\end{pmatrix}
\succeq 0 \, ,
\\[.9em]
&&  Y,Z \in \Pos(\W\otimes \V) \, ,
\\[.2em]
&& \fnorm{\A(X)-y} \leq \eta \, .
\end{array}
\label{eq:nnorm_reconstruction}
\end{equation}
What is more, several of these recovery guarantees can be established using the geometric proof techniques presented in Section~\ref{sub:convex_recovery}.
For results established that way, combining Observation~\ref{obs:smaller_DC} with Corollary~\ref{cor:subset} allows us to draw the following conclusion.

\begin{implication}[Inheriting recovery guarantees]\label{imp:Inheriting}
	For bipartite operators $X_0 \in \L(\W \otimes \V)$ that satisfy $\nnorm{X_0} = \jnorm{X_0}$, any recovery guarantee for nuclear norm minimization,
	which is based on the nuclear norm's descent cone, also holds for square norm minimization.
\end{implication}
 
This implication 
indicates that replacing nuclear norm regularization \eqref{eq:Mast} by
\begin{equation}\label{eq:Mdiamond}
X^\mysquare_\eta \coloneqq \argmin\{ \jnorm{X}: \fnorm{\A(X)-y} \leq \eta \} 
\end{equation}
results in an estimation procedure that performs at least as well whenever $\jnorm{X_0} = \nnorm{X_0}$. 
In fact, Observation~\ref{obs:smaller_DC} suggests that it may actually outperform traditional recovery procedures. 
Also, the SDP formulation for the square norm \cite{Wat12} allows one to recast the optimization 
\eqref{eq:Mdiamond} as
\begin{equation}
\label{eq:jnorm_reconstruction}
\arraycolsep=1.5pt
\begin{array}{r l l}
X^\mysquare_\eta
=&\underset{X,Y,Z}{\argmin} &\hspace{-2ex} 
  \frac{\dim(\V)}{2}\bigl( \snorm{\Tr_\W(Y)} + \snorm{\Tr_\W(Z)} \bigr)
\\[.3em]
&\st &
	\begin{pmatrix}
	Y & -X \\ -X\ad & Z
	\end{pmatrix}
\succeq 0 \, ,
\\[.9em]
&& Y,Z \in \Pos(\W\otimes \V) \, ,
\\[.2em]
&&  \fnorm{\A(X)-y} \leq \eta \, ,
\end{array}
\end{equation}
which, just like the optimization \eqref{eq:nnorm_reconstruction}, is a convex optimization problem that can be solved computationally efficiently and also practically using standard software such as CVX \cite{cvx,GraBoy08}.
In the remainder of this section, we present three measurement scenarios for which Implication~\ref{imp:Inheriting} holds. 
The first one is a version of Ref.\ \cite[Example 4.4]{Tro14} 
which is valid for reconstructing real-valued matrices. 
In its original formulation with nuclear norm minimization, it follows from
combining Proposition~\ref{prop:general_reconstruction} and Eq.~\eqref{eq:lambda_nnorm_Gaussian}. 

\begin{proposition}[Stable recovery of real matrices via Gaussian measurements]
	\label{prop:Gaussian_CS}
	Let $X_0 \in \L (\W \otimes \V)$ be a real valued, bipartite matrix of rank $r$
	that obeys $\jnorm{X_0} = \nnorm{X_0}$. 
	Also, suppose that each measurement matrix $A_i$ is a real-valued standard Gaussian matrix and the overall noise is bounded as $\fnorm{\epsilon}\leq \eta$. 
	Then, $m \geq C r \dim (\W \otimes \V)$ noisy measurements of the form \eqref{eq:measurement_operator} suffice to guarantee
	\begin{equation}
	\fnorm{X^\mysquare_\eta - X_0} \leq \frac{C' \eta}{\sqrt{m}}
	\label{eq:recovery_bound}
	\end{equation}
	with probability at least $1-\e^{-C'' m}$. Here, $C$, $C'$ and $C''$ denote  absolute constants.
\end{proposition}

With high probability (w.h.p.), this statement assures \emph{stable} recovery, meaning that the reconstruction error \eqref{eq:recovery_bound} scales linearly in the noise bound $\eta$ and inversely proportional to $\sqrt{m}$. 

For the sake of clarity, we have refrained from providing explicit values for the constants $C,C'$ and $C''$ in Proposition~\ref{prop:Gaussian_CS}. 
However, resorting to Tropp's bound \eqref{eq:lambda_nnorm_Gaussian} on the minimal conical eigenvalue of a Gaussian sampling matrix reveals that 
stably recovering any rank-$r$ matrix obeying Eq.~\eqref{eq:equality} requires roughly
\begin{equation}
m \gtrsim 6 r\, ( \dim (\V) \dim (\W) - r)
\end{equation}
independently selected Gaussian measurements.

Proposition~\ref{prop:Gaussian_CS} is a prime example for a \emph{non-uniform} recovery guarantee: 
For any fixed rank-$r$ matrix $X_0$ obeying Eq.~\eqref{eq:equality}, 
$m$ randomly chosen measurements of the form \eqref{eq:measurement} suffice to stably reconstruct $X_0$ w.h.p. 
For some measurement scenarios, stronger recovery guarantees can be established. Called \emph{uniform} recovery guarantees, these results assure that one choice of sufficiently many random measurements suffices w.h.p.\ to reconstruct all possible matrices of a given rank. 

A uniform recovery statement can be established for the following real-valued measurement scenario \cite{KabKueRau15}: suppose that with respect to an arbitrary orthonormal basis of $\W \otimes \V$, each matrix element of $A_i$ is an independent instance of a real-valued random variable $a$ obeying
\begin{equation}
\EE \left[ a \right] = 0, \quad
\EE \left[ a^2 \right] = 1, \quad \textrm{and} \quad
\EE \left[ a^4 \right] \leq F, \label{eq:fourth_moments}
\end{equation}
where $F \geq 1$ is an arbitrary constant. Measurement matrices of this form can be considered as a generalization of Gaussian measurement matrices, where each matrix element corresponds to a standard Gaussian random variable. 
In Ref.\ \cite{KabKueRau15} -- see also Refs.\ \cite{KabRauTer15a,KabRauTer15b} -- a uniform recovery guarantee for such measurement matrices has been established by means of the \emph{Frobenius robust rank null space property} \cite[Definition 10]{KabKueRau15}.
Such a proof technique is different from the geometric one introduced in Section~\ref{sub:convex_recovery}. However, as laid out in the appendix, some auxiliary statements allow for reassembling technical statements from these works
to yield a slightly weaker, but still uniform, statement by means of analyzing descent cones. Implication~\ref{imp:Inheriting} is applicable for such a result and yields the following. 

\begin{proposition}[Stable, uniform recovery of real matrices via measurement matrices with finite fourth moments] \label{prop:fourth_moments}
Consider the measurement process described in Eq.~\eqref{eq:measurement_operator}, 
where each $A_i \in \L (\W \otimes \V)$ is an independent random matrix of the form \eqref{eq:fourth_moments}. Fix $r\geq 1$ 
and suppose that $m \geq C_F\, r \dim (\W \otimes \V)$. 
Then, w.h.p., every real-valued matrix $X_0 \in \L (\V \otimes \W)$ of rank at most $r$ and
obeying $\nnorm{X} = \snorm{X}$ can be stably reconstructed from the measurements  \eqref{eq:measurement_operator} by means of square norm minimization \eqref{eq:Mdiamond}. Here, $C_F$ is a constant that only depends on the fourth-moment bound $F$.
\end{proposition}

This is a uniform recovery guarantee that assures stability towards additive noise corruption $\epsilon$ in the measurement process \eqref{eq:measurement_operator}.
However, it does not establish robustness towards the model assumption of low rank. For nuclear norm minimization, the main results in \cite{KabKueRau15} employ a null space property argument that  does cover this additional stability aspect.

We conclude this section with two uniform recovery guarantees for Hermitian low-rank matrices from measurement matrices $A_i$ that are proportional to rank-one projectors, i.e., $A_i = a_i a_i^\ast$ for some $a_i \in \W \otimes \V$. 
Originally established for nuclear norm minimization in Ref.\ \cite{KueRauTer15}, by using an extension of the geometric proof techniques presented in Section~\ref{sub:convex_recovery}, Implication~\ref{imp:Inheriting} is directly applicable to such measurements. 

\begin{proposition}[Stable, uniform recovery of Hermitian matrices from rank-one measurements] \label{prop:rank_one}
	Consider recovery of Hermitian rank-$r$ matrices $X_0 \in \L (\W
	\otimes \V)$ that obey $\jnorm{X_0}=\nnorm{X_0}$
	from rank-one
	measurements of the form $A_i = a_i a_i^\ast$. 
	Let $n = \dim(\W \otimes \V)$. 
	Then stable and uniform recovery guarantees for square norm
	minimization \eqref{eq:Mdiamond} analogous to
	Proposition~\ref{prop:fourth_moments} hold if either
	\begin{enumerate}
		\item 
		the measurements $a_i$ are 
		$m \geq C_G r n$ random Gaussian vectors in 
		$\W \otimes \V$ or
		\item
		the measurements $a_i$ are 
		$m \geq C_{4D} r n \log (n)$ vectors
		drawn uniformly from a 
		complex projective $4$-design (see Eq.~\eqref{eq:def_design} below).
	\end{enumerate}
	Once more, $C_G$ and $C_{4D}$
	denote absolute constants of sufficient size.
\end{proposition}

In the statement above, a \emph{complex projective $t$-design} is a
configuration of vectors $\{w_j\}_{j \in [k]}$ which is ``evenly distributed'' on a sphere
in the sense that sampling uniformly from it reproduces the moments of the
Haar measure up to order $2t$
\cite{DelGoeSei77,RenBluRob04,AmbEme07}.
More precisely,
\begin{equation}\label{eq:def_design}
	\frac{1}{k}\sum_{j=1}^k (w_j w_j^\dagger)^{\otimes t} =
	\int_{\fnorm{w}=1} (w w^\dagger)^{\otimes t} \mathrm{d}w.
\end{equation}
The second statement in Proposition~\ref{prop:rank_one} can be seen as
``partial derandomization'' of the first one
\cite{GroKraKue15_partial}.

\section{Application to the recovery of linear maps on operators}
\label{sec:application_maps}
Now we come to three concrete applications concerning linear maps that take operators in $\L(\V)$ to operators in $\L(\W)$. Our reconstruction based on the square norm can be applied to such maps by identifying them with operators in $\L(\W \otimes \V)$. 
We start with introducing some relevant notation and explain such an identification, the Choi-Jamio{\l}kowski isomorphism, in more detail. 
Then we present numerical results on retrieval of certain unitary basis changes, quantum process tomography, and blind matrix deconvolution. 

\subsection{Notation concerning linear maps on operators} 
\label{sec:notation_maps}
Our square norm is closely related to the diamond norm, which is defined for linear operators $M: \L (\V) \to \L(\W)$ that map operators
to operators. We call such objects \emph{maps} and denote their space by 
$\M (\V,\W)\coloneqq \L\left( \L(\V) \to \L(\W) \right)$, 
or simply by 
$\M \left(\V \right) \coloneqq \M \left( \V,\V \right)$. 
We also denote maps by capital Latin letters. Concretely, for 
$M \in \M (\V,\W)$ and $X \in \L (\V)$ we write $M(X) \in\L(\W)$. 
A particularly simple example is the identity map
$\1_{\L(\V)} \in \M(\V)$ which obeys $\1_{\L(\V)} (X) = X$ for all $X\in \L(\V)$. 

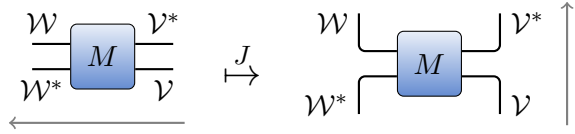
\begin{figure}
	\centering
	%%% ---------------- tikzpicture ------------------
	\leavevmode
	\beginpgfgraphicnamed{fig3}%
	\definecolor{niceblue}{rgb}{0.33,0.5,0.8}%
	\tikzset{%
	    sbox/.style = {draw, rounded corners = .5ex,%
			   minimum height = 1.8\baselineskip,%
			   minimum width = 2.2em},%standard box
	    blau/.style = {top color=niceblue!12,%
			   bottom color=niceblue!90},%blue shading
	    Bbox/.style = {sbox, blau},% blue box
	    leg/.style = {rounded corners = .5ex,thick},%tensor legs
	    dir/.style = {gray,thick}%direction
	    }%
	\begin{tikzpicture}[node distance = 1ex]
	  \node (M1) {
	    \begin{tikzpicture}
	    \Mbox{$M$}
	    \draw [leg] (More) -- ++(\len,0) node (Xo) [near end, above] {$\V^\ast$};
	    \draw [leg] (Mure) -- ++(\len,0) node (Xu) [near end, below] {$\V$};
	    \draw [leg] (Moli) -- ++(-\len,0) node (Yo) [near end, above] {$\W$};
	    \draw [leg] (Muli) -- ++(-\len,0) node (Yu) [near end, below] {$\W^\ast$};
	    \draw [dir, ->] (Xu.south east) ++(0,-1ex) coordinate (ure) -- (ure-|Yu.west);
	    \end{tikzpicture}    
	    };
	  \node (J) [right = of M1] {\Large{$\overset{J}{\mapsto}$}};
	  \node (M2) [right = of J]{
	    \begin{tikzpicture}
	    \Mbox{$M$}
	    \draw [leg] (More) -- ++(\len,0) --++(0,\len) node (Xo) [near end, right] {$\V^\ast$};
	    \draw [leg] (Mure) -- ++(\len,0) --++(0,-\len) node (Xu) [near end, right] {$\V$};
	    \draw [leg] (Moli) -- ++(-\len,0) --++(0,\len) node (Yo) [near end, left] {$\W$};
	    \draw [leg] (Muli) -- ++(-\len,0) --++(0,-\len) node (Yu) [near end, left] {$\W^\ast$};
	    \draw [dir, <-] (Xo.north east) ++(1ex,0) coordinate (ore) -- (ore|-Xu.south);
	    \end{tikzpicture}
	    };
	%   \node (TrY) [right = of M2]{\Large{$\overset{\Tr_\W}{\mapsto}$}};
	%   \node (M3) [right = of TrY]{
	%     \begin{tikzpicture}
	%     \Mbox{$M$}
	%     \draw [leg] (Moli) -- ++(-\len,0) |- (Muli);
	%     \draw [leg] (More) -- ++(\len,0) --++(0,\len) node (Xo) [near end, right] {$\V^\ast$};
	%     \draw [leg] (Mure) -- ++(\len,0) --++(0,-\len) node (Xu) [near end, right] {$\V$};
	%     \draw [dir, ->] (Xo.north east) ++(1ex,0) coordinate (ore) -- (ore|-Xu.south);
	%     \end{tikzpicture}
	%   };
	\end{tikzpicture}
	\endpgfgraphicnamed
	%%% ------------------------------------
	%%% compile with:
	%%% pdflatex --jobname=fig3 diamonds.tex	
	% \includegraphics{fig3}% include pre-compiled file
	%%% ------------------------------------
	\caption{Tensor network diagrams: tensors are denoted by boxes with one line for each index. Contraction of two indices corresponds to connection of the corresponding lines. 
		\newline
		\capstr{Left:} Order-$4$ tensor $M$ as a map from $\L(\V)\cong \V\otimes \V^\ast$ to $\L(\W) \cong \W\otimes \W^\ast$. \newline
		\capstr{Right:} Its Choi-matrix $J(M)$ as an operator on $\W^\ast \otimes \V\cong \W \otimes \V$. 
	}
	\label{fig:tensor_diagram}
\end{figure}

We would like to identify maps in $\M(\V,\W)$ with operators in $\L(\W\otimes \V)$, for which we have discussed certain reconstruction schemes. For this purpose, we employ a very useful isomorphism, called the
\emph{Choi-Jamio{\l}kowski isomorphism} \cite{Cho75,Jam72}. In order to explicitly define this isomorphism, we fix an orthogonal basis $(e_i)$ of $\V$. This also gives rise to an operator basis 
\begin{equation}\label{eq:operator_basis}
E_{i,j}\coloneqq e_i e_j^T \in \L(\V)  
\end{equation}
and we define \emph{vectorization} $\vec : \L(\V)\to \V\otimes \V$ by the linear extension of 
\begin{equation}
  \vec(E_{i,j}) \coloneqq e_i\otimes e_j \, .
\end{equation}
Then the Choi-Jamio{\l}kowski isomorphism $J$ is defined by
\begin{equation}
\begin{aligned}
J: \M (\V,\W) 
&\to \L ( \W \otimes \V) \\
M  
&\mapsto \sum_{i,j=1}^{\dim(\V)} M(E_{i,j}) \otimes E_{i,j} \label{eq:choi} \, .
\end{aligned}
\end{equation}
The resulting operator $J(M)$ is called the \emph{Choi matrix} of $M$. 
It can be straightforwardly checked that Eq.~\eqref{eq:choi} is equivalent to setting
\begin{equation}\label{eq:Jdef}
J(M) = \Bigl( M \otimes \1_{\L(\V)} \Bigr) \bigl( \vec(\1_\V) \vec(\1_\V)^T \bigr) .
\end{equation}
In Appendix~\ref{sec:appendixChoi}, we provide a basis independent definition of the Choi-Jamio{\l}kowski isomorphism, which is an instance of the natural isomorphism 
$\W \otimes\W^\ast\otimes \V^\ast \otimes \V \cong \W \otimes \V^\ast \otimes\W^\ast \otimes \V$. 
This identification is illustrated in Figure~\ref{fig:tensor_diagram}.

Let us also explicitly mention the case where the underlying vector spaces are again written as $\V = \CC^n$ and $\W = \CC^N$. 
Then $M(X)$ is given by an $N\times N$ matrix with elements 
$M(X)_{i,j} = \sum_{k,l=1}^n M_{i,j,k,l} X_{k,l}$. 
Here, the map $M$ is represented as an array 
$(M_{i,j,k,l})_{i,j\in [N],\, k,l \in [n]}$. 
The Choi matrix is obtained by swapping the second and third index and then taking the joint first two and joint last two indices as first and second index of the Choi matrix, respectively, i.e., 
$J(M)_{(i,k),(j,l)} = M_{i,j,k,l}$ with $i,j\in [N]$ and $k,l \in [n]$. 	

Similarly to the definition of the spectral norm \eqref{eq:snrom_def}, the nuclear norms on $\L(\V)$ and $\L(\W)$ induce a norm on $\M(\V,\W)$,  
\begin{equation}
\norm{M}_{\ast \to \ast} \coloneqq 
\sup_{X \in \L(\V)} \frac{ \nnorm{M (X)}}{\nnorm{X}} \, . \label{eq:induced}
\end{equation}
Perhaps surprisingly, the induced nuclear norm of maps of the
form $M\otimes \1_{\L(\V)}$ can be computed efficiently
\cite{Wat09,BenTaS09, Wat12}, as explained in detail below. 
This motivates studying the
\emph{diamond norm} \cite{KitSheVya02}
\begin{equation}
	\dnorm{M} \coloneqq \norm{M \otimes \1_{\L(\V)}}_{\ast \to \ast}.  \label{eq:diamond}
\end{equation}
It plays an important role in quantum mechanics \cite{KitSheVya02} and
is also the core concept of this work.  Using the
Choi-Jamio{\l}kowski isomorphism, the diamond norm \eqref{eq:diamond}
can indeed be written \cite{Wat12} as
\begin{equation}\label{eq:dnorm_variational}
 \dnorm{M} = \frac{\jnorm{J(M)}}{\dim(\V)} \, ,
\end{equation}
where the square norm was defined variationally in Eq.~\eqref{eq:jnorm_variational}. Hence, for the case of a measurement map $\A : \M(\V,\W) \to \CC^m$, the reconstruction based on the square norm \eqref{eq:jnorm_reconstruction} can also be written as
\begin{equation}
\label{eq:dnorm_reconstruction}
\arraycolsep=1.5pt
\begin{array}{r l l}
 M_\eta^\diamond
 =&\argmin & \tkw 2 \snorm{\Tr_\W(Y)} + \tkw 2 \snorm{\Tr_\W(Z)}
 \\[.3em]
  &\st\ &
	\begin{pmatrix}
          Y & -J(M) \\ -J(M)\ad & Z
        \end{pmatrix}
             \succeq 0 \, ,
  \\[.9em]
  &&  Y,Z \in \Pos(\W\otimes \V) \, ,
\\[.2em]
  && \fnorm{\A(M)-y} \leq \eta \, .
\end{array}
\end{equation} 
In our numerical experiments we solve this minimization problem using CVX \cite{cvx,GraBoy08}.

\begin{figure}
  \centering
	\includegraphics[width = .5\linewidth]{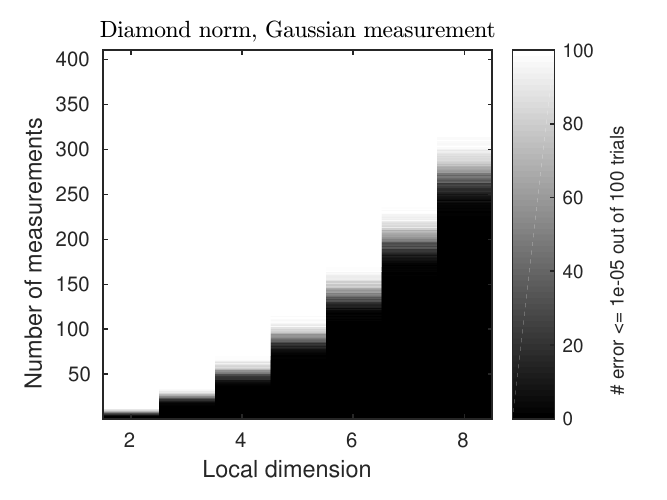}%
    \includegraphics[width = .5\linewidth]{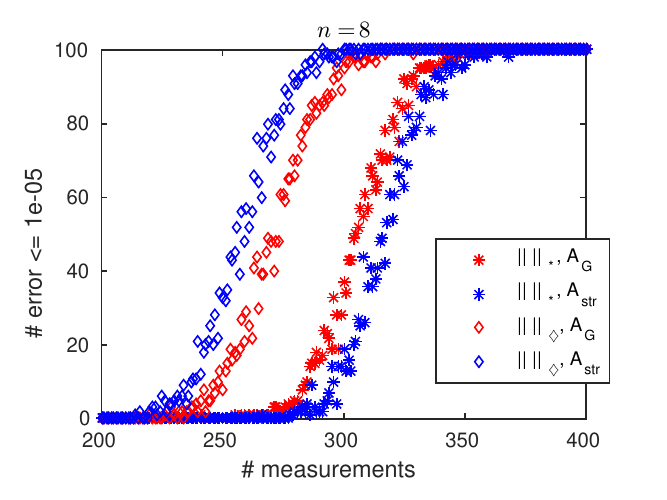}%
  \caption{Retrieval of $M(X) = UXV$ for the case of real numbers and $\ev=0$. 
	   $U,V \in\O(n)$ are orthogonal matrices drawn from the Haar measure in each trial. 
	   The plots show the number of trails out of $100$ with small errors, $\fnormn{M^\diamond_\eps-M_0} \leq 10^{-5}$ and $\fnormn{M^\ast_\eps-M_0} \leq 10^{-5}$, respectively, and with $\eta$ chosen as machine precision $\eps$. 
	   \newline
    \capstr{Left:} Diamond norm minimization with Gaussian measurements for different local dimensions~$n$. \newline
    \capstr{Right:} Comparison of diamond norm and nuclear norm with Gaussian and structured measurements.
  Note that the structured measurements improve the reconstruction based on the diamond norm while for the reconstruction based on the nuclear norm Gaussian measurements turn out to work better. The computation time needed for the recovery is approximately the same for both methods. }
  \label{fig:UV}
\end{figure}

\subsection{Retrieval of certain unitary basis changes}\label{sec:UV}
Our problem of retrieval of unitary basis changes is motivated by the phase retrieval problem. Retrieving phases from measurements that are ignorant towards them has a long-standing history in various scientific disciplines \cite{Wal63}. 
A discretized version of this problem can be phrased as the task of inferring a complex vector $x \in \CC^n$ from measurements of the form
\begin{equation}\label{eq:phaseless_measurements}
y_i = \left| \langle a_i, x_i \rangle \right|^2, 
\end{equation}
where $a_1,\ldots,a_m \in \CC^n$. 
Recently, the mathematical structure of this problem has received considerable attention
\cite{Wal63,CanEldStr11,CanLi13,CanStroVor13,AleBanFicMix14,CanLiSol14,GroKraKue15_partial,GroKraKue15_masked}. 
One way of approaching this problem is to recast it as a matrix problem which has the benefit that the measurements \eqref{eq:phaseless_measurements} become linear. Indeed, setting $X \coloneqq x x\ad$ and $A_i = a_i a_i\ad$ reveals that
\begin{equation}
y_i = \left| \langle a_i, x \rangle \right|^2 
= 
\Tr \bigl( a_i a_i\ad x x\ad \bigr) 
= 
\Tr \left( A_i X \right).
\end{equation}
This ``lifting'' trick allows for re-casting the phase retrieval problem as the task of recovering a Hermitian rank-one matrix $X = x x\ad$ from linear measurements of the form $A_i = a_i a_i\ad$. 

Recently, Ling and Strohmer \cite{LinStr15} used similar techniques to recast the important problem of self-calibration in hardware devices 
as the task to recover a non-Hermitian rank-one matrix $X = x y\ad$ from similar linear measurements.

In this section, we consider the matrix-analogue of such a task and set $\V = \CC^n = \W$ but keep $\V$ and $\W$ as labels.
Concretely, we consider maps $M \in \M(\V,\W)$ of the form
\begin{equation}\label{eq:MeqUXV}
  M(X) = U X V \, ,
\end{equation}
where $U$ and $V$ are fixed unitaries. Note that any such map has a Choi matrix of the form
\begin{equation}
\begin{aligned}
J (M) &= \bigl[M \otimes \1_{\L(\V)}\bigr]\bigl( \vec (\1_{\V} ) \vec (\1_{\V} )\ad \bigr)
\\
&= \left(  U \otimes \1_{\V}\vec (\1_{\V} ) \right) 
% \\& \qquad \quad  \times
  \left(  V \otimes \1_{\V}\vec (\1_{\V} ) \right)\ad 
\\
&= \vec(U)\vec(V)\ad  ,
\end{aligned}
\end{equation}
which corresponds to an outer product of the form $x y\ad$.
Moreover, unitarity of both $U$ and $V$ assures that all such maps meet the requirements of Theorem \ref{thm:extremality}.

We aim to numerically recover such maps from two different types of measurements:
(i) Gaussian measurements and (ii) structured measurements. The Gaussian measurements are given by a measurement map $\AG : \M(\V,\W) \to \CC^m$ with real and imaginary parts of all of its components drawn from a normal distribution with zero mean and unit variance. 
In the case of structured measurements, $M$ receives \mbox{rank-$1$} inputs $x_j y_j\ad$ and then inner products with regular measurement matrices $A_j$ are measured. 
More precisely, the measurement map $\Astr : \M(\V,\W) \to \CC^m$ is given by 
\begin{align}
  \Astr(M)_j &\coloneqq \Tr\bigl( A_j M(x_j y_j\ad)\bigr) \, , \quad j\in [m]\, ,
\end{align}
where $x_j,y_j$ are chosen uniformly from the complex unit sphere $\left\{ z \in \V: \; \fnorm{z}=1 \right\} \subset \V$.
The random matrices $A_j$ are independently distributed as the random matrix $U_A D V_A$, where $D \in \L (\V)$ is fixed as a real-valued diagonal matrix and both $U_A$ and $V_A$ are chosen independently from the Haar measure over $U(\dim(\V))$. 
For our numerical studies, we mostly restrict ourselves to even dimensions $n= \dim(\V)$ and set
$D = \frac{2}{n}(1,-1,2,-2,\ldots,n/2,-n/2)$. 
This in particular assures $\snorm{D}=1$.
As we will see, similar types of measurements can be used in quantum process tomography and blind matrix deconvolution. 

For both measurement setups, we find that diamond norm reconstruction outperforms nuclear norm reconstruction; see Figure~\ref{fig:UV}. Interestingly, the structured measurements are better than the Gaussian measurements for the diamond norm reconstruction, while for the nuclear norm reconstruction we find the converse. 

Finally, we would like to to point out that Ling and Strohmer introduced a new algorithm -- dubbed ``SparseLift'' -- to efficiently reconstruct the signals they consider and simultaneously promote sparsity \cite{LinStr15}. 
It is an intriguing open problem to compare the performance of SparseLift to the constrained diamond norm minimization advocated here for different types of practically relevant measurement ensembles. We leave this idea to future work. 

\subsection{Quantum process tomography}\label{sec:quantum}
The problem of reconstructing quantum mechanical processes from measurements is referred to as \emph{quantum process tomography}. 
As explained in the next paragraph, quantum processes are described by
maps that saturate the norm inequality \eqref{eq:equality} and thus are natural candidates for diamond norm-based methods.

\begin{figure}
  \centering
  \includegraphics[width = .48\linewidth]{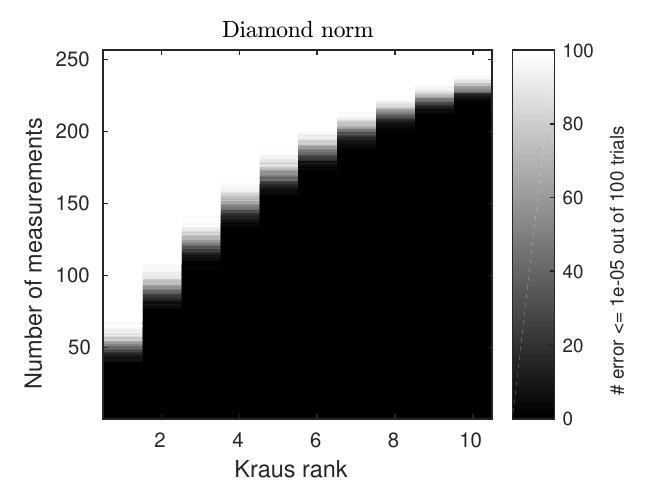}
  \includegraphics[width = .48\linewidth]{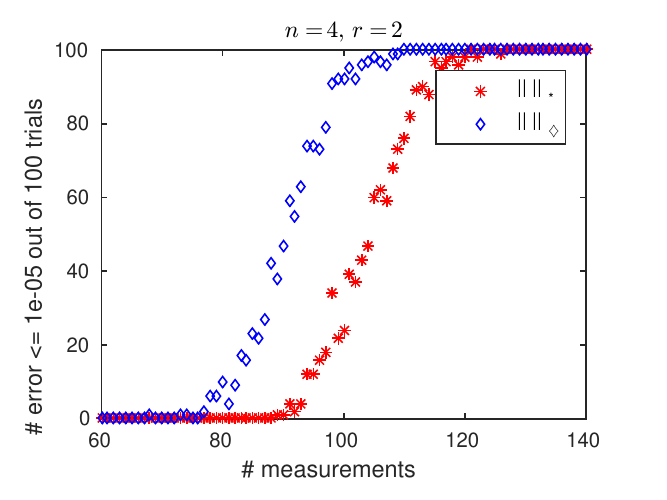}
  \caption{Retrieval of random quantum channels $M_0$ acting on two qubits ($n=4$) with $\ev=0$. 
	   The plots show the number of trails out of $100$ with small errors, $\fnorm{M^\diamond_\eps-M_0} \leq 10^{-5}$ and $\fnorm{M^\ast_\eps-M_0} \leq 10^{-5}$, respectively, and with $\eta$ chosen as machine precision $\eps$. 
	   In each trial, $M_0$ is constructed from the Haar measure on $U \in \U(r\, \dim(\V))$ by tracing out an $r$ dimensional space. 
	   \newline
    \capstr{Left:} Diamond norm recovery for different Kraus ranks. \newline
    \capstr{Right:} Comparison of diamond norm and trace norm of the Choi matrix for Kraus rank $r=2$. The diamond norm recovery works with fewer measurements than the conventional nuclear norm recovery, while the computation time is approximately the same.}\label{fig:qm}
\end{figure}

In the following paragraph we briefly outline the mathematical essentials of (finite dimensional) quantum mechanics in general, and quantum process tomography in particular. 
We content ourselves with introducing the key concepts and defer the interested reader to the book~\cite{NieChu10} for a thorough introduction to quantum mechanics from a computer scientist's perspective.

\paragraph*{Preliminaries} 
A positive semidefinite operator $\rho \in \Pos(\V)$ with unit trace $\Tr(\rho)=\nnorm{\rho}=1$ is called a \emph{density operator} and a matrix representation is a \emph{density matrix}. The convex space of density operators is denoted by $\DM(\V) \subset \Pos(\V)$ and its elements are referred to as \emph{quantum states}. 
The extreme elements of $\DM(\V)$ are called \emph{pure states} and are given by rank-one operators of the form $\psi \psi\ad$ with $2$-norm normalized \emph{state vectors} $\psi \in \V$. An \emph{observable} is a self-adjoint operator $A \in \Herm(\V)$ and the \emph{expectation value} of $A$ in state $\rho \in \DM(\V)$ is $\Tr(\rho\,A)$. Note that in the case where $\rho$ and $A$ are diagonal, $\rho$ corresponds to a classical probability vector and $A$ to a random variable also with expectation value $\Tr(\rho\, A)$. For the following definitions it is helpful to know that quantum systems are composed to larger quantum systems by taking tensor products of operators. 
A map $M \in \M(\V,\W)$ is called \emph{completely positive} if $J(M)\in \Pos(\W\otimes \V)$ with $J$ from Eq.~\eqref{eq:Jdef}. This is the case if and only if for every vector space $\V$ the map $M \otimes \1_{\L(\V)}$ preserves the cone $\Pos(\V\otimes \W)$ of positive semidefinite operators. $M\in \M(\V,\W)$ is called \emph{trace preserving} if $\Tr(M(X)) = \Tr(X)$ for all $X \in \L(\V)$. The convex space of maps that are both, completely positive and trace preserving is denoted by $\CPT(\V,\W) \subset \M(\V,\W)$ and its elements are quantum operations as they map density operators to density operators and they are also called \emph{quantum channels}. 
Importantly, any $M \in \CPT(\V,\W)$ satisfies $\dnorm{M}=1$ and $\nnorm{J(M)}=\dim(\V)$. Due to the relation \eqref{eq:dnorm_variational} between the diamond and square norm $J(M)$ automatically fulfils the extremality condition \eqref{eq:equality}.

The \emph{Kraus rank} of a quantum channel $M \in \CPT(\V,\W)$ is the rank of its Choi matrix $J(M)$. A channel $M \in \CPT(\V,\W)$ of Kraus rank $r$ can be written as 
\begin{equation}\label{eq:Krausrep}
  M(\rho) = \sum_{j=1}^r K_j \rho K_j\ad \, ,
\end{equation}
where $K_j \in \L(\V \to \W)$ are so-called \emph{Kraus operators} satisfying $\sum_{j=1}^r K_j\ad K_j = \1_{\V}$, 
and no other such decomposition has fewer terms. A special role is
played by \emph{unitary channels}, which are channels of unit Kraus
rank. In this case, the single Kraus operator in the Kraus representation \eqref{eq:Krausrep}
has to be unitary. Unitary quantum channels describe coherent
operations in the sense that for isolated quantum systems (i.e.,
systems that are decoupled from anything else) one can only have
unitary quantum channels. Quantum channels describing situations where
the system is affected by noise have Kraus ranks larger than one. In
many experimental situations, one aims at the implementation of a
unitary channel, but actually implements a channel whose Kraus rank is
larger than one, but is still approximately low. Therefore, process
tomography of quantum channels with low Kraus rank is an important
task in quantum experiments. Also, in the context of \emph{quantum
error correction}, low-rank deviations turn out to have a particularly
adverse impact \cite{KueLonFla15}. 
This underscores the need to
design efficient estimation protocols for this case.

In the next paragraph, we present numerical results showing that,
indeed, replacing the nuclear norm with the diamond norm in ``compressive process tomography'' improves the
results. We expect that using the diamond norm as a ``drop
in replacement'' for the nuclear norm will also lead to improvements
in other, more involved process tomography schemes. For example,
Kimmel and Liu \cite{KimLiu15} combine compressed process tomography
with ideas from \emph{randomized benchmarking}
\cite{EmeAliZyc05,GraFerCor14}. This combination allows recovery using
only Clifford measurements that are robust to state preparation and
measurement (SPAM) errors. Their recovery guarantees are based on the
geometric arguments presented in Section~\ref{sub:convex_recovery} and allow for measurements drawn from unitary $2$-designs. It thus seems fruitful to also investigate the diamond norm in their setting. 

\paragraph*{Numerical results for quantum process tomography}
The task is to reconstruct $M_0\in \CPT(\V,\W)$ from measurements of the form
\begin{equation}
y = \A (M_0)+ \ev,
\end{equation}
where $\A: \M (\V,\W)\to \RR^m$ encodes linear data acquisition, 
 $y \in \RR^m$ summarizes the measurement outcomes, and $\ev \in \RR^m$
 represents additive noise. 
The most general measurements conceivable in this context are so-called \emph{process POVMs} \cite{Zim08}. However, here we consider the case where $\A$ is given by the preparation of
pure states given by state vectors $\psi_j \in \V$ and measurements
of observables $A_j \in \Herm(\W)$, where $j\in[m]$. 
This yields similar measurements as in Section~\ref{sec:UV}, 
\begin{equation}
  y_j = \A(M_0)_j \coloneqq \Tr\bigl(A_j M_0(\psi_j\psi_j\ad)\bigr) + \ev_j \, , \quad j \in [m] \, , 
\end{equation}
where each $\psi_j \in \V$ is chosen uniformly and independently from the complex unit
sphere in $\V$.
Each observable $A_j \in \Herm (\W)$ 
is of the form $A_j = U_j D U_j\ad$, where each $U_j \in U(\dim(\W))$ is drawn independently from the Haar measure over all unitaries.
Once more, $D \in \Herm (\W)$ 
is a fixed Hermitian operator. 
With this measurement setup, quantum channels can be recovered from
few measurements. Once more, diamond norm reconstruction outperforms the
conventional nuclear norm reconstruction, see Figure~\ref{fig:qm}. 

\begin{figure*}
	  \centering
    \includegraphics[width = .48\linewidth]{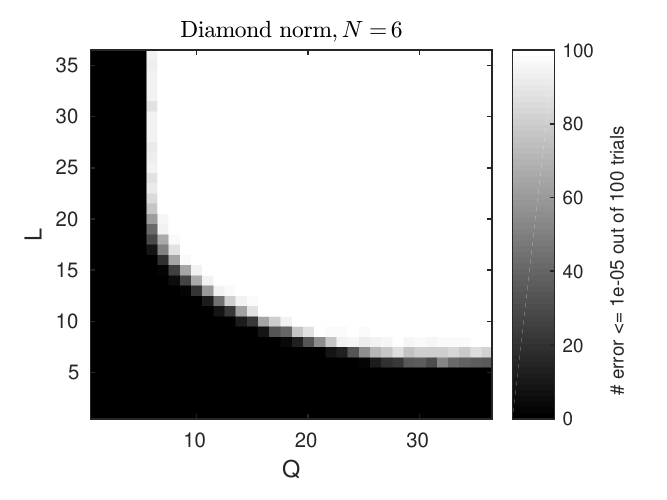}
	\includegraphics[width = .48\linewidth]{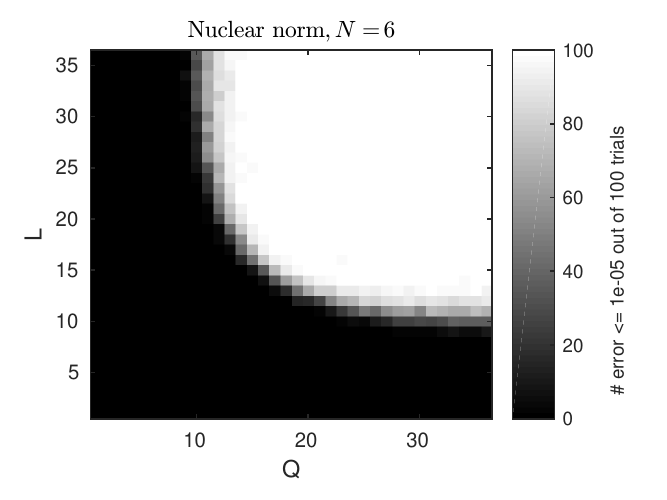}
	\caption{Blind matrix deconvolution with $N=6$ and $\ev=0$. 
		The plots show the number of trials out of $100$ with small errors, $\fnorm{M^\diamond_\eps-M_0} \leq 10^{-5}$ and $\fnorm{M^\ast_\eps-M_0} \leq 10^{-5}$, respectively, and with $\eta$ chosen as machine precision $\eps$. 
		In each trial, $M_0$ is constructed from Haar-random unitaries $U$ and $V$.
		\newline
		\capstr{Left:} Recovery via diamond norm. %\\
		\capstr{Right:} Recovery via nuclear norm.
		\newline
		The diamond norm recovery works with fewer measurements than the nuclear norm recovery, while the computation time is approximately the same.}
	\label{fig:deconvolution}
\end{figure*}

\subsection{Blind matrix deconvolution}\label{sec:deconv}
The \emph{blind deconvolution scheme} as considered in Ref.\ \cite{AhmRecRom12} 
aims to reconstruct unknown vectors $h\in \RR^k$ and $m\in \RR^n$. 
From these vectors, length $L$ signals are being generated as
\begin{equation}%\label{eq:Romberg_Bh}
	w = Bh \qquad\text{and}\qquad
	x = Cm,
\end{equation}
for known $B\in \L \left(\RR^k \to \RR^L \right)$ and $C\in \L \left( \RR^n \to \RR^L \right)$. The observed quantity is the circular convolution of $w$ and $x$,
\begin{equation}
	y = w* x = 
	\sum_{i=1}^L \left( \sum_{j=1}^L w_j\, x_{i-j+1 \mod L}\right)e_i \, ,
\end{equation}
where $(e_1,\ldots,e_L)$ denotes the standard basis of $\RR^L$.
This gives rise to  a bi-linear problem, which can still be solved using a lifting technique to a variant of the matrix completion problem.

The type of problem considered in this work allows for the \emph{blind matrix deconvolution}, in which not vectors $h,w$, but orthogonal or unitary matrices $U,V$ reflecting unknown rotations are reconstructed. 

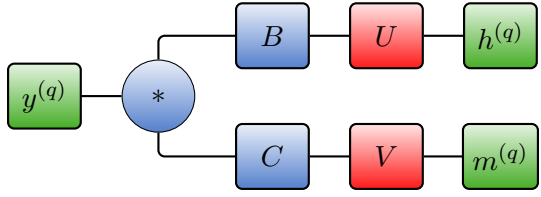
\begin{figure}%[b]
	\centering%
%%% ---------------- tikzpicture ------------------
	\leavevmode
	\beginpgfgraphicnamed{fig7}%	
	\definecolor{niceblue}{rgb}{0.33,0.5,0.8}%
	\colorlet{mygreen}{OliveGreen!90!blue}%
	\tikzset{%
		sbox/.style = {draw, rounded corners = .5ex,%
			minimum height = 1.8\baselineskip,%
			minimum width = 2.5em},%standard box
		blau/.style = {top color=niceblue!12,%
			bottom color=niceblue},%blue shading
		rot/.style = {top color=red!10,%
			bottom color=red!90},%shading
		grau/.style = {top color=gray!12,%
			bottom color=gray!90},%shading
		gruen/.style = {top color=mygreen!10,%
			bottom color=mygreen!70},%green shading
		starbox/.style = {draw, circle, blau, minimum height = 2\baselineskip},% blue box
		Bbox/.style = {sbox,blau},%
		Gbox/.style = {sbox,gruen},%
		Rbox/.style = {sbox,rot},%
		leg/.style = {rounded corners = .5ex,thick},%tensor legs
		dir/.style = {gray,thick}%direction
	}%
	\def\ydist{.8}%
	\def\xdist{1.5}%
\begin{tikzpicture}%
	\node (star) [starbox] at (0,0) {$\ast$};
	\draw [leg] (star) -- ++ (-\xdist,0) node (y) [Gbox] {$y^{(q)}$};
	\draw [leg] (star) |- ++ (\xdist,\ydist) node (B) [Bbox] {$B$};
	\draw [leg] (star) |- ++ (\xdist,-\ydist) node (C) [Bbox] {$C$};
	\draw [leg] (B) -- ++(\xdist,0) node (U) [Rbox]{$U$};
	\draw [leg] (C) -- ++(\xdist,0) node (V) [Rbox]{$V$};
	\draw [leg] (U) -- ++(\xdist,0) node (h) [Gbox]{$h^{(q)}$};
	\draw [leg] (V) -- ++(\xdist,0) node (m) [Gbox]{$m^{(q)}$};
\end{tikzpicture}%
\endpgfgraphicnamed
%%% ------------------------------------
% \includegraphics{fig7}% include pre-compiled file
%%% ------------------------------------  
\caption{Blind matrix deconvolution: Measurement vectors in green, fixed operations in blue, and unknown signal in red. }
\label{fig:def_deconv}
\end{figure}

In this new problem, for known $B,C\in \L \left(\CC^N \to \CC^L \right)$
and real vectors $h^{(q)},m^{(q)}\in \RR^N$
with $q \in [Q]$, that are an input to the problem, we seek to reconstruct $U,V\in U(n)$ 
from the circular convolutions $y^{(q)} = w^{(q)}* x^{(q)}$ of $w^{(q)}$ and $x^{(q)}$, where now
\begin{equation}
\begin{aligned}
	w^{(q)} &= B U h^{(q)},
	\\
	x^{(q)} &= CV m^{(q)},
\end{aligned}
\end{equation}
see also Figure~\ref{fig:def_deconv}.
The observations are given by the $Q$ vectors $y^{(q)} = w^{(q)} \ast x^{(q)}$ or, equivalently, by
\begin{equation}
\begin{aligned}
	\hat y^{(q)} &= \hat w ^{(q)}\circ \hat x^{(q)}
	\\ 
	&= (FBUh^{(q)}) \circ (FCVm^{(q)}) \, ,
\end{aligned}
\end{equation}
where $F_{j,k} \coloneqq \e^{2\pi \i \, (j-1)(k-1)/N}/\sqrt{N}$ defines the Fourier transform $F$ and $(a \circ b)_j \coloneqq a_j b_j$ the Hadamard product of vector $a$ and $b$. Let us denote the $j$-th rows of $FB$ and $FC$ by $\hat b_j$ and $\hat c_j$, respectively. Then 
\begin{equation}
 y_j^{(q)} 
 = \hat b_j^T Uh^{(q)}\, \hat c_j V m^{(q)}
 =
 \Tr(E_j\, U \rho^{(q)} V^T) 
\end{equation}
with the unit rank matrices 
$E_j \coloneqq \hat c_j^T \hat b_j$ and $\rho^{(q)} \coloneqq h^{(q)} m^{(q)^T}$. 

Indeed, this is precisely a problem of the form discussed here, 
\begin{equation}
y_l^{(q)} = \langle E_l\ad, M(\rho^{(q)} )\rangle 
\end{equation}
with $\V = \W=\CC^n$ and
\begin{equation}
	M(X) = U X \overline V.
\end{equation}
Up to a complex phase factor, $U$ and $V$ can be trivially reconstructed from $M$. 
That is to say, a matrix version of blind deconvolution can readily be cast into the
form of problems considered in this work. 
Numerically, we find a recovery from few samples and that the diamond norm reconstruction outperforms the nuclear norm based reconstruction from Ref.\ \cite{AhmRecRom12} adapted to our setting; see Figure~\ref{fig:deconvolution}. 
Many practical applications of this problem are conceivable: The reconstruction of an unknown drift of a polarization degree of freedom in a channel problem is only one of the many natural ramifications of this setup. 

%%% ==================================================
%%% ================== Proofs ===========================
%%% ==================================================
\section{Proofs}
\label{sec:proofs}
In this section, we prove Proposition~\ref{prop:bounds_on_diamond_norm} and an extension of Theorem \ref{thm:extremality}. 
In order to do so, we first define a generalization of the sign matrix to matrices that are not necessarily Hermitian. This will give rise to the left and right absolute values of arbitrary matrices. Then we introduce SDPs, complementary slackness, and state the SDP for the square norm in standard form. Combining all these concepts, this section cumulates in the proofs of Proposition~\ref{prop:bounds_on_diamond_norm} and Theorem~\ref{thm:extremality}. 

\paragraph{Remark:}
After publication of this work a direct proof of Theorem \ref{thm:extremality}, which is not based on SDPs, has been found \cite{MicKliKue16}.

\subsection{Auxiliary statements}
The singular value decomposition of a matrix $X \in \L\left( \CC^n \right)$ is 
\begin{equation}
X = U \Sigma V\ad, \label{eq:svd}
\end{equation}
where $U,V \in \U(n)$ are unitaries and $\Sigma\in \Pos(\CC^n)$ is positive-semidefinite and diagonal. This decomposition allows one to define a ``sign matrix'' of $X$:

\begin{definition}[Sign matrix] \label{def:sgn}
For any matrix $X \in \L \left( \CC^n \right)$ with singular value decomposition \eqref{eq:svd} we define its \emph{sign matrix} to be $S_X \coloneqq V U\ad$.
\end{definition}

Note that the sign matrix is in general not unique, but always unitary and it obeys
\begin{align}
X S_X &= 
U \Sigma U\ad = \sqrt{X X\ad},   \label{eq:sgnaux1}\\
X\ad S_X\ad &= 
V \Sigma V\ad = \sqrt{X\ad X}. \label{eq:sgnaux2} 
\end{align}
Therefore, $S_X$ indeed generalizes the sign-matrix $\sign(X)$ (which is defined exclusively for Hermitian matrices).

The following auxiliary statement will be required later on and follows from a Schur complement rule.

\begin{lemma}\label{lem:block_matrix_psd1}
For every $A\in \L(\V\to \W)$, one has
\begin{equation}
\begin{pmatrix}
\snorm{A} \1_\W & \pm A \\
\pm A\ad & \snorm{A} \1_\V
\end{pmatrix} \succeq 0  \, .\label{eq:block_matrix_psd1} 
\end{equation}
\end{lemma}

\subsection{Semidefinite programming}\label{sec:SDPs}
Semidefinite programs (SDPs) are a class of optimization problems that can be evaluated efficiently with standard software, e.g. by using CVX \cite{cvx,GraBoy08}. 

\begin{definition}[Semidefinite program]\label{def:SDP}
A \emph{semidefinite program} is specified by a triple $\left( \Xi, C,D \right)$, where $C \in \Herm(\V)$ and $D \in \Herm(\W)$ are self-adjoint operators and 
$\Xi: \; \L(\V) \to \L(\W)$
is a hermiticity preserving linear map.
With such a triple, one associates a pair of optimization problems:
\begin{alignat}{2}
{\normalfont \textbf{Primal:}}\qquad&
  \maximize &\quad &\!\! \Tr(C Z ) \label{eq:primal}
  \\
  &\st &&  \Xi (Z) = D\, , \label{eq:Xi(X)exD}
  \\
  &&& Z \in \Pos(\V) \, ,\label{eq:Xsucceq0}
  \\
{\normalfont \textbf{Dual:}}\qquad &
  \minimize && \!\! \Tr(D Y ) \label{eq:dual}
  \\
  &\st && \Xi\ad (Y) \succeq C \, , \label{eq:Xiad(Y)succeqC}
  \\
  &&&  Y \in \Herm(\W) \, . 
\qquad 
\end{alignat}
$Z^\sharp\in \Herm(\V)$ is called \emph{primal feasible} if it satisfies Eq.~\eqref{eq:Xi(X)exD} and Eq.~\eqref{eq:Xsucceq0}. It is called \emph{optimal primal feasible} if, additionally, for $Z=Z^\sharp$ in Eq.~\eqref{eq:primal} the maximum is attained. Similarly, $Y^\sharp \in \Herm(\W)$ is called \emph{dual feasible} if it satisfies Eq.~\eqref{eq:Xiad(Y)succeqC} and \emph{optimal dual feasible}, if for $Y = Y^\sharp$ the minimum in Eq.~\eqref{eq:dual} is attained. 
\end{definition}

SDPs that exactly reproduce the problem structure outlined in this definition are said to be in \emph{standard form}. But for specific SDPs, equivalent formulations might often be more handy. 

\emph{Weak duality} refers to the fact that the value of the primal SDP cannot be larger than the value of the dual SDP, i.e., that $\Tr(CZ) \leq \Tr(DY)$ for any primal feasible point $Z$ and dual feasible point $Y$. 
An SDP is said to satisfy \emph{strong duality} if the optimal values coincide, i.e., if for some optimal primal feasible and dual feasible points $Z^\sharp$ and $Y^\sharp$ it holds that $\Tr(CZ^\sharp) = \Tr(DY^\sharp)$. 
In fact, from a weak condition, called Slater's condition, strong duality follows. 

\begin{lemma}[Complementary slackness, see, e.g., Ref.~\cite{Bar02}]
\label{lem:comp_slack}
Suppose that $(\Xi,C,D)$ characterizes an SDP that obeys strong duality and let 
$Z^\sharp \in \Herm(\V)$ and $Y^\sharp \in \Herm(\W)$ denote optimal primal and dual feasible points, respectively (i.e. $\Tr\left(C Z^\sharp \right) = \Tr\left( D Y^\sharp \right)$).
Then
\begin{equation}
\Xi\ad(Y^\sharp) Z^\sharp 
= 
C Z^\sharp \quad \textrm{and} \quad \Xi (Z^\sharp) Y^\sharp = D Y^\sharp.
\end{equation}
\end{lemma}

The following, somewhat exhaustive, classification of the square norm's SDP will be instrumental later on.

\begin{lemma}[Watrous' SDP for the diamond norm in standard \cite{Wat12}] 
\label{lem:Watrous_SDP}
\hfill
\label{thm:diamond}
Let $X \in \L(\W \otimes \V)$ be a bipartite operator. 
Then its square norm $\jnorm{X}$ can be evaluated by means of an SDP $(\Xi,C,D)$ that satisfies strong duality. In standard form, it is given by the block-wise defined matrices
\begin{align}\label{eq:Watrous_C}
C &=
\frac{\dim(\V)}{2}
%\begin{pmatrix}
%\zerom_\V & \zerom_\V & \zerom\ad & \zerom\ad \\
%\zerom_\V & \zerom_\V & \zerom\ad & \zerom\ad \\
%\zerom & \zerom & \zerom_{\W \otimes \V} & X \\
%\zerom & \zerom &  X\ad & \zerom_{\W \otimes \V}  \\
%\end{pmatrix}
\begin{pmatrix}
0 & 0 & 0 & 0 \\
0 & 0 & 0 & 0 \\
0 & 0 & 0 & X \\
0 & 0 &  X\ad & 0  \\
\end{pmatrix}
 % \\ &
\in
\Herm \left( \V \oplus \V \oplus (\W \otimes \V) \oplus (\W \otimes \V) \right) ,
% \nonumber
\\
D &= 
%\begin{pmatrix}
%1 & 0  &  \zerov_{\W \otimes \V}\ad &  \zerov_{\W \otimes \V}\ad\\
%0 & 1 & \zerov_{\W \otimes \V}\ad &  \zerov_{\W \otimes \V}\ad \\
%\zerov_{\W \otimes \V} & \zerov_{\W \otimes \V} & \zerom_{\W \otimes \V} & \zerom_{\W \otimes \V} \\
%\zerov_{\W \otimes \V} & \zerov_{\W \otimes \V} & \zerom_{\W \otimes \V} & \zerom_{\W \otimes \V} 
%\end{pmatrix}
\begin{pmatrix}
1 & 0  &  0 &  0\\
0 & 1 & 0 & 0 \\
0 & 0 & 0 & 0 \\
0 & 0 & 0 & 0 
\end{pmatrix}
 % \\ &
\in 
\Herm \left( \CC \oplus \CC \oplus (\W \otimes \V) \oplus (\W \otimes \V) \right) \, ,
%\nonumber
\end{align}
%where $\zerov_{\W \otimes \V} \in \W \otimes \V$ denotes the zero-vector, and 
%$\zerom \in \L(\V \to \W \otimes \V)$, as well as
%$\zerom_\V\in \L(\V)$ represent zero matrices of appropriate dimension. 
where the zeros are zero vectors or operators of appropriate dimension. 
Finally, the map  
\begin{equation} 
 \begin{split}
\Xi: \Herm\bigl( (\V \oplus \V \oplus (\W \otimes \V) \oplus (\W \otimes \V ) \bigr)
 % \\
\to
\Herm \bigl( \CC \oplus \CC \oplus (\W \otimes \V ) \oplus (\W \otimes \V) \bigr)
 \end{split}
\end{equation}
acts as
\begin{align}
\Xi &
\begin{pmatrix}
W_0 & \cdot & \cdot & \cdot 
\\
\cdot & W_1 & \cdot & \cdot 
\\
\cdot & \cdot & Z_0 & \cdot \\
\cdot & \cdot & \cdot & Z_1
\end{pmatrix}
= 
 % \\& \quad 
%\begin{pmatrix}
%\Tr(W_0) & 0  &  \zerov_{\W \otimes \V}\ad &  \zerov_{\W \otimes \V}\ad\\
%0 & \Tr(W_1) & \zerov_{\W \otimes \V}\ad &  \zerov_{\W \otimes \V}\ad \\
%\zerov_{\W \otimes \V} & \zerov_{\W \otimes \V} & Z_0 - \1_\W \otimes W_0& \zerom_{\W \otimes \V} \\
%\zerov_{\W \otimes \V} & \zerov_{\W \otimes \V} & \zerom_{\W \otimes \V} & Z_1 - \1_\W \otimes W_1
%\end{pmatrix} 
\begin{pmatrix}
\Tr(W_0) & 0  &  0 &  0\\
0 & \Tr(W_1) & 0 &  0 \\
0 & 0 & Z_0 - \1_\W \otimes W_0& 0 \\
0 & 0 & 0 & Z_1 - \1_\W \otimes W_1
\end{pmatrix} %\nonumber
\end{align}
and has an adjoint map given by
\begin{align}\label{eq:Xi_ad}
%\begin{split}
\Xi\ad &
\begin{pmatrix}
\lambda_0 & \cdot & \cdot & \cdot \\
\cdot & \lambda_1 & \cdot & \cdot \\
\cdot & \cdot & Y_0 & \cdot \\
\cdot & \cdot & \cdot & Y_1
\end{pmatrix}
=
 % \\& \quad
% \small
%\begin{pmatrix}
%\lambda_0 \1_\V{} - \Tr_{\W} \left( Y_0 \right)
%  & \zerom_\V 
%  & \zerom\ad
%  & \zerom\ad
%\\
%\zerom_\V & \lambda_1 \1_\V - \Tr_{\W} \left( Y_1 \right)& \zerom\ad & \zerom\ad
%\\
%\zerom & \zerom & Y_0 &  \zerom_{\W \otimes \V}  \\
%\zerom & \zerom &  \zerom_{\W \otimes \V} & Y_1
%\\
%\end{pmatrix} ,
\begin{pmatrix}
\lambda_0 \1_\V{} - \Tr_{\W} \left( Y_0 \right)
& 0
& 0
& 0
\\
0 & \lambda_1 \1_\V - \Tr_{\W} \left( Y_1 \right)&0 & 0
\\
0 & 0 & Y_0 &  0  \\
0 & 0 &  0 & Y_1
\\
\end{pmatrix} .
% \nonumber
%\end{split}
\end{align}
%where $\zerom \in \L(\V \to \W \otimes \V)$, once more, represents a zero-matrix.
\end{lemma}

Lemma~\ref{lem:Watrous_SDP} presents an SDP for the square norm in standard form. Although this standard form is going to be important for our proofs, it is somewhat unwieldy. 
Fortunately, elementary modifications 
allow \cite{Wat12} to reduce the SDP to the following pair.
\begin{align}
\label{eq:Watrous_primal}
&
\begin{array}{rll}
	\multicolumn{1}{l}{\textbf{Primal:}}&&
	\\ 
	 \jnorm{X} =& \max & \kw 2 \Tr(X Z) + \kw 2\Tr( X\ad Z\ad)
	\\[.5em]
	&\st & 
		\begin{pmatrix}
		\1_\W \otimes \rho & Z \\ Z\ad & \1_\W \otimes \sigma
		\end{pmatrix}
	\succeq 0\, , \hspace{1cm} \quad
	\\[1em]
	&& \Tr(\rho) = \Tr(\sigma) = \dim(\V) \, ,
	\\[.2em]
	&& \rho,\sigma \in \Pos(\V)\, ,
	\\[.2em]
	&& Z \in \L(\W\otimes \V)   \,   ,
\end{array}
%\intertext{and}
\\[1em]
%%%-------------------------------------------------------
\label{eq:Watrous_SDP_dual}
&
\begin{array}{rll}
	\multicolumn{1}{l}{\textbf{Dual:}}&&
	\\
	 \jnorm{X}
	=& \min & \frac{\dim(\V)}{2}\Bigl(\snorm{\Tr_\W(Y)} + \snorm{\Tr_\W(Z)}\Bigr)
	\\[.5em]
	&  \st &  
		\begin{pmatrix}
		Y & -X \\ -X\ad & Z
		\end{pmatrix}
	\succeq 0 \, ,
	\\[1em]
	&& Y,Z \in \Pos(\W\otimes \V) \, .
\end{array}
\end{align}
This simplified SDP pair for the square norm comes in handy for establishing the final claim in Proposition~\ref{prop:bounds_on_diamond_norm}. 
For Hermitian matrices, the first two bounds presented there were already established in Ref.\ 
\cite[Lemma 7]{WalFla14}. Here, we show that an analogous strategy remains valid for matrices that need not be Hermitian.

\begin{proof}[Proof of Proposition~\ref{prop:bounds_on_diamond_norm}]
Let us start with recalling the variational definition \eqref{eq:jnorm_variational} of the square norm:
\begin{equation}\label{eq:jnorm_def2}
 % \begin{split}
 \jnorm{X} \coloneqq
 \max\{\nnorm{(\1_\W \otimes A)X(\1_\W \otimes B)}:\
 	 % \\&
	   A,B \in \L(\V), \, 
 	 % \\&
	   \fnorm{A}=\fnorm{B} = \sqrt{\dim(\V)} \}.
 % \end{split}
\end{equation}
As already mentioned, inserting $A=B=\1$ into Eq.~\eqref{eq:jnorm_def2} establishes the lower bound \eqref{eq:nnorm_leq_jnorm} ($\nnorm{X} \leq \jnorm{X}$).

A generalized version of H\"older's inequality and $\snorm{\1_\W \otimes A} = \snorm{\1_\W} \snorm{A} \leq \fnorm{A}$ assures
\begin{align}
 \nnorm{(\1_\W \otimes A)X(\1_\W \otimes B)}
% \\
\leq &
\snorm{\1_\W \otimes A} \nnorm{X} \snorm{\1_\W \otimes B} \nonumber
\\
\leq &
\fnorm{A} \fnorm{B} \nnorm{X}
\end{align}
for any $A,B \in \L (\V)$ and $X \in \L (\W \otimes \V)$.
Inserting this bound into Eq.~\eqref{eq:jnorm_def2} results in
$\jnorm{X} \leq \dim (\V)\nnorm{X}$, which is the second bound \eqref{eq:jnorm_leq_nnorm}. 

The final bound ($\jnorm{X} \leq \dim (\W \otimes \V) \snorm{X}$) can be proved similarly using a generalized version of H\"older's inequality. 
However, in order to demonstrate the usefulness of SDPs in this context, we provide a different proof. 
For this purpose, we consider the simplified version of the square norm's dual SDP \eqref{eq:Watrous_SDP_dual}. 
Lemma~\ref{lem:block_matrix_psd1} assures that setting $Y = Z = \snorm{X} \1_{\W \otimes \V}$ results in a feasible point of this program.
Inserting this point into the objective function yields a value of $\dim (\W) \dim (\V) \snorm{X}$, because $\snorm{ \Tr_\W \left( \1_{W \otimes \V} \right)} = \snorm{ \dim (\W) \1_\V} = \dim (\W)$.
The bound follows from this value and the structure of the optimization problem \eqref{eq:Watrous_SDP_dual}.
\end{proof}

\subsection{Proof of Theorem~\ref{thm:extremality}}
In this section, we prove an extension of Theorem~\ref{thm:extremality}. In particular, this more general result relates Theorem~\ref{thm:extremality} to optimal feasible points in Watrous' SDP from Lemma~\ref{lem:Watrous_SDP}. 
These will contain the generalizations of the sign matrix from Definition~\ref{def:sgn}.

\newcommand{\Zsharp}{
	\begin{pmatrix}
		\1_{\V} & 0 & 0 & 0 \\
		0 & \1_\V & 0 & 0 \\
		0 & 0 & \1_{\W \otimes \V} & S_{X}\ad \\
		0 & 0 &  S_{X} & \1_{\W \otimes \V}  \\
	\end{pmatrix}
}

\begin{theorem}[Extremal operators as optimal feasible points] \label{thm:extremal}
 Let $X \in \L(\W\otimes \V)$ be a bipartite operator and set $n\coloneqq \dim(\V)$. 
 Then the points \ref{item:extremal}--\ref{item:reduced_eq} are equivalent:\begin{enumerate}[label=\roman*)]
  \item \label{item:extremal}
  $X$ satisfies
  \begin{equation}\label{eq:extremality_assumption}
    \jnorm{X} =  \nnorm{X},
  \end{equation}
  \item
  \label{item:Xsharp}
  Some $Z^\sharp\in \Herm((\V\oplus \V \oplus (\W \otimes \V) \oplus (\W \otimes \V))$ of the form
  \begin{equation}
    Z^\sharp \coloneqq \kw{n} \Zsharp
\end{equation}
is a primal optimal feasible point for Watrous' SDP $\left( \Xi, C,D \right)$ from Lemma~\ref{lem:Watrous_SDP}.
 \item 
 \label{item:Ysharp}
  Some $Y^\sharp\in \Herm \left( \CC \oplus \CC \oplus (\W \otimes \V ) \oplus (\W \oplus \V) \right)$ of the form
  \begin{equation}
  \label{eq:Y_sharp}
%   \hspace{-2.5em}
  Y^\sharp
  = 
  \frac{1}{2}\begin{pmatrix}
  \nnorm{X} & \cdot & \cdot & \cdot \\
  \cdot &  \nnorm{X} & \cdot & \cdot \\
  \cdot & \cdot & n \, \sqrt{ X X\ad} & \cdot \\
  \cdot & \cdot & \cdot & n \, \sqrt{ X\ad X}
  \end{pmatrix}
  \end{equation}
  is a dual optimal feasible point for Watrous' SDP $\left( \Xi, C,D \right)$ from Lemma~\ref{lem:Watrous_SDP}.

 \item 
 \label{item:reduced_propto}
 $X$ satisfies 
 \begin{equation}
   \Tr_{\W} \left( \sqrt{ X X\ad} \right)
    \propto 
    \Tr_{\W} \left( \sqrt{ X\ad X} \right)
    \propto
     \1_{\V} \, .
 \end{equation}
 
 \item 
 \label{item:reduced_eq}
 $X$ satisfies 
 \begin{equation}
	 \label{eq:equality_proof}
   \Tr_{\W} \left( \sqrt{ X X\ad} \right)
    = 
    \Tr_{\W} \left( \sqrt{ X\ad X} \right)
    = 
    \frac{\nnorm{X}}{n} \,  \1_{\V} \, .
 \end{equation}
\end{enumerate}
\end{theorem}

Similar to the actual SDP, the optimal feasible points presented in Theorem~\ref{thm:extremal} have simplified counterparts that correspond to optimal feasible points of the simplified SDPs \eqref{eq:Watrous_primal} and \eqref{eq:Watrous_SDP_dual}. For the sake of completeness, we present them in the following corollary.

\begin{corollary}
For any $X \in \L (\W \otimes \V)$,
optimal feasible points of the primal \eqref{eq:Watrous_primal} and the dual SDP \eqref{eq:Watrous_SDP_dual} for the square norm are given by the following.
\begin{align}
\text{Primal optimal feasible point:}  & \quad Z = S_X, \rho = \sigma = \1_\V \, ,
% \quad \text{and} \qquad\qquad
\\
\text{Dual optimal feasible point:}  & \quad Y = \sqrt{X X\ad}, \; Z = \sqrt{X\ad X} \, .
\end{align}
\end{corollary}

This statement follows straightforwardly from Theorem~\ref{thm:extremal} by considering the reduced formulations \eqref{eq:Watrous_primal} and \eqref{eq:Watrous_SDP_dual} of the SDP from Lemma~\ref{lem:Watrous_SDP}. 

\begin{proof}[Proof of Theorem~\ref{thm:extremal}]
For $X=0$ all statements are evident. From now on, we assume that $X\neq 0$. 

% \setlist[description]%{font=\normalfont\itshape}
% \begin{description}%[leftmargin=-3ex]
%%% -------------------------- item -----------------------------------
\begin{proof}[Proof of \ref{item:extremal} $\Rightarrow$ \ref{item:Xsharp}]
Note that $Z^\sharp \succeq 0$ by Lemma~\ref{lem:block_matrix_psd1}. 
Straightforward evaluation of $\Xi (Z^\sharp)$ from Lemma~\ref{lem:Watrous_SDP} reveals that $Z^\sharp$ is indeed a primal feasible point:
\begin{equation*}
\begin{aligned}
\Xi \bigl( Z^\sharp \bigr)
&= 
%\\ &\hspace{-1cm}
%\scriptsize
%\begin{pmatrix}
%\kw{n}\Tr(\1_{\V}) & 0  
%&  \zerov_{\W \otimes \V}\ad &  \zerov_{\W \otimes \V}\ad
%\\
%0 & \kw{n}\Tr(\1_{\V}) & \zerov_{\W \otimes \V}\ad &  \zerov_{\W \otimes \V}\ad 
%\\
%\zerov_{\W \otimes \V} & \zerov_{\W \otimes \V} & 
%  \kw{n} \1_{\W \otimes \V} - \1_\W \otimes \kw{n} \1_\V& \zerom_{\W \otimes \V} 
%\\
%\zerov_{\W \otimes \V} & \zerov_{\W \otimes \V} & \zerom_{\W \otimes \V} 
%  & \kw{n}\1_{\W \otimes \V} - \1_\W \otimes \kw{n} \1_\V
%\end{pmatrix} 
% \footnotesize
\kw{n}
%\begin{pmatrix}
%\Tr(\1_{\V}) & 0  
%&  \zerov_{\W \otimes \V}\ad &  \zerov_{\W \otimes \V}\ad
%\\
%0 & \Tr(\1_{\V}) & \zerov_{\W \otimes \V}\ad &  \zerov_{\W \otimes \V}\ad 
%\\
%\zerov_{\W \otimes \V} & \zerov_{\W \otimes \V} & 
% \1_{\W \otimes \V} - \1_\W \otimes  \1_\V& \zerom_{\W \otimes \V} 
%\\
%\zerov_{\W \otimes \V} & \zerov_{\W \otimes \V} & \zerom_{\W \otimes \V} 
%& \1_{\W \otimes \V} - \1_\W \otimes  \1_\V
%\end{pmatrix} 
\begin{pmatrix}
\Tr(\1_{\V}) & 0  
&  0 &  0
\\
0 & \Tr(\1_{\V}) & 0 &  0 
\\
0 & 0 & 
\1_{\W \otimes \V} - \1_\W \otimes  \1_\V& 0 
\\
0 & 0 & 0
& \1_{\W \otimes \V} - \1_\W \otimes  \1_\V
\end{pmatrix} 
\\
&=
%\begin{pmatrix}
%1 & 0  &  \zerov_{\W \otimes \W}\ad &  \zerov_{\W \otimes \W}\ad\\
%0 & 1 & \zerov_{\W \otimes \V}\ad &  \zerov_{\W \otimes \V}\ad \\
%\zerov_{\W \otimes \V} & \zerov_{\W \otimes \V} & \zerom_{\W \otimes \V}  & \zerom_{\W \otimes \V} \\
%\zerov_{\W \otimes \V} & \zerov_{\W \otimes \V} & \zerom_{\W \otimes \V} & \zerom_{\W \otimes \V}
%\end{pmatrix}  
\begin{pmatrix}
1 & 0  & 0 &  0\\
0 & 1 & 0 & 0 \\
0 & 0 & 0  & 0 \\
0 & 0 & 0 & 0
\end{pmatrix}  
= D \, .
\end{aligned}
\end{equation*}
In order to show optimality,  
we evaluate the primal SDP's objective function given by $C$ in Eq.~\eqref{eq:Watrous_C}. Employing formulas \eqref{eq:sgnaux1} and \eqref{eq:sgnaux2} to express the absolute values of $X$, we obtain
\begin{equation}
\begin{aligned}
% &\phantom{={.}}
\Tr \left( C Z^\sharp \right)
% \\
&= 
\frac{n}{2}\Tr \Bigl[
%\begin{pmatrix}
%\zerom_\V & \zerom_\V & \zerom\ad & \zerom\ad \\
%\zerom_\V & \zerom_\V & \zerom\ad & \zerom\ad \\
%\zerom & \zerom & \zerom_{\W \otimes \V} & X \\
%\zerom & \zerom &  X\ad & \zerom_{\W \otimes \V}  \\
%\end{pmatrix}
\begin{pmatrix}
	0 & 0 & 0 & 0 \\
	0 & 0 & 0 & 0 \\
	0 & 0 & 0 & X \\
	0 & 0 &  X\ad & 0  \\
\end{pmatrix}
 % \\&\qquad \qquad\quad\
\cdot\frac{1}{n}\Zsharp \Bigr]
\\
&= \frac{1}{2} \left( \Tr \left( X S_{X} \right)
+ \Tr \left(X\ad S_{X}\ad \right) \right) \\
&= \frac{1}{2} \left( \Tr \left( \sqrt{X X\ad} \right) + \Tr \left( \sqrt{X\ad X}\right)  \right) \\
&= \frac{1}{2} \left( \nnorm{X} + \nnorm{X} \right)
= \nnorm{X} \, .
\end{aligned}
\end{equation}
By assumption \eqref{eq:extremality_assumption}, this is indeed optimal.
\end{proof}

\begin{proof}[Proof of \ref{item:Xsharp} $\Rightarrow$ \ref{item:Ysharp} and \ref{item:reduced_propto}]
\renewcommand{\qedsymbol}{}
Strong duality of Watrous' SDP from Lemma~\ref{lem:Watrous_SDP} assures that an optimal dual solution $Y^\sharp$ exists and that complementary slackness holds. Since $\Xi\ad$ from Eq.~\eqref{eq:Xi_ad} does not depend on block off-diagonal terms, optimal feasibility only depends on the block diagonal parts. Hence, we write $Y^\sharp$ as 
\begin{equation}
Y^\sharp
=
\begin{pmatrix}
\lambda_0 & \cdot & \cdot & \cdot \\
\cdot & \lambda_1 & \cdot & \cdot \\
\cdot & \cdot & Y_0 & \cdot \\
\cdot & \cdot & \cdot & Y_1
\end{pmatrix} .
\end{equation}
Complementary slackness (Lemma~\ref{lem:comp_slack}) implies that 
\begin{equation*}
\begin{aligned}
 % &\phantom{={.}}
\Xi\ad \bigl( Y^\sharp \bigr) Z^\sharp
 % \\ 
&
% \small
=
\frac{1}{n}
%\begin{pmatrix}
%\lambda_0 \1_\V - \Tr_{\W} \left( Y_0 \right)& \zerom_\V & \zerom\ad & \zerom\ad \\
%\zerom_\V & \lambda_1 \1_\V - \Tr_{\W} \left( Y_1 \right)& \zerom\ad & \zerom\ad \\
%\zerom & \zerom & Y_0 &  \zerom_{\W \otimes \V}  \\
%\zerom & \zerom &  \zerom_{\W \otimes \V} & Y_1\\
%\end{pmatrix}
\begin{pmatrix}
\lambda_0 \1_\V - \Tr_{\W} \left( Y_0 \right)& 0 & 0 & 0 \\
0 & \lambda_1 \1_\V - \Tr_{\W} \left( Y_1 \right)& 0 & 0 \\
0 & 0 & Y_0 &  0  \\
0 & 0 &  0 & Y_1\\
\end{pmatrix}
% \\
% &
% \small
% \hspace{4.3cm} 
\cdot 
\Zsharp
\\
&
% \small
=
\frac{1}{n}
%\begin{pmatrix}
%\lambda_0 \1_\V - \Tr_{\W} \left( Y_0 \right)& \zerom_\V & \zerom\ad & \zerom\ad \\
%\zerom_\V & \lambda_1 \1_\V - \Tr_{\W} \left( Y_1 \right)& \zerom\ad & \zerom\ad \\
%\zerom & \zerom & Y_0 &  Y_0 \; S_{X}\ad \\
%\zerom & \zerom &  Y_1 \; S_{X} & Y_1\\
%\end{pmatrix} \label{eq:slackness1}
\begin{pmatrix}
\lambda_0 \1_\V - \Tr_{\W} \left( Y_0 \right)& 0 & 0 & 0 \\
0 & \lambda_1 \1_\V - \Tr_{\W} \left( Y_1 \right)& 0 & 0 \\
0 & 0 & Y_0 &  Y_0 \; S_{X}\ad \\
0 & 0 &  Y_1 \; S_{X} & Y_1\\
\end{pmatrix} \label{eq:slackness1}
\end{aligned}
\end{equation*}
and 
\begin{equation}
C Z^\sharp
=
\frac{1}{2}
%\begin{pmatrix}
%\zerom_\V & \zerom_\V & \zerom\ad & \zerom\ad \\
%\zerom_\V & \zerom_\V & \zerom\ad & \zerom\ad \\
%\zerom & \zerom & X S_{X} & X \\
%\zerom & \zerom &  X\ad & X\ad S_{X}\ad\\
%\end{pmatrix}  \label{eq:slackness2}
\begin{pmatrix}
	0 & 0 & 0 & 0 \\
	0 & 0 & 0 & 0 \\
	0 & 0 & X S_{X} & X \\
	0 & 0 &  X\ad & X\ad S_{X}\ad\\
\end{pmatrix}  \label{eq:slackness2}
\end{equation}
must equal each other. This in turn demands that 
\begin{align}
Y_0 &= \frac{n}{2} X S_{X} = \frac{n}{2} \sqrt{ X X\ad} \quad \textrm{as well as} \\
Y_1 &= \frac{n}{2} X\ad S_{X}\ad = \frac{n}{2} \sqrt{X\ad X},
\end{align}
where we have once more employed identities \eqref{eq:sgnaux1} and \eqref{eq:sgnaux2} for $S_X$ to obtain the absolute values of $X$. Equality of \eqref{eq:slackness1} and \eqref{eq:slackness2} in the first two diagonal entries (also guaranteed by complementary slackness) furthermore assures
\begin{align}
\lambda_0 \1_{\V} - \Tr_\W \left( Y_0 \right)
&= \lambda_0 \1_{\V} - \frac{n}{2} \Tr_\W \left( \sqrt{X X\ad} \right)
% \nonumber\\
%&= \zerom_\V 
% &
= 0 
\intertext{and}
\lambda_1 \1_{\V} - \Tr_\W \left( Y_1 \right)
&= \lambda_0 \1_{\V} - \frac{n}{2} \Tr_\W \left( \sqrt{X\ad X} \right)
% \nonumber\\
%&= \zerom_\V. 
% &
= 0. 
\end{align}
Hence, 
\begin{align}
  \lambda_0\,n &= \frac{n}{2} \nnorm{\Tr_\W(Y_0)} =\frac{n}{2} \nnorm{X}
  \quad \textrm{and}
  \\
  \lambda_1\,n &= \frac{n}{2} \nnorm{\Tr_\W(Y_1)} = \frac{n}{2} \nnorm{X} 
\end{align}
and both, \ref{item:Ysharp} and \ref{item:reduced_propto} follow. 
\end{proof}
%%% -------------------------- item -----------------------------------
\begin{proof}[Proof of \ref{item:reduced_propto} $\Rightarrow$ \ref{item:reduced_eq}]
\renewcommand{\qedsymbol}{}
Let $c_1, c_2 >0$ be constants such that
\begin{align}
\Tr_\W \left( \sqrt{ X X\ad} \right) 
&= c_1 \1_\V \quad \textrm{and} \label{eq:constant_aux1} \\
\Tr_\W \left( \sqrt{ X\ad X} \right) \label{eq:constant_aux2}
&= c_2 \1_\V.
\end{align}
Taking the trace of both equations and recognizing the nuclear norm reveals that
\begin{equation}
\begin{aligned}
  \nnorm{X} &= \Tr\Bigl(\sqrt{XX\ad} \Bigr) 
  \\
  &= \Tr \left( \Tr_\W \left( \sqrt{XX\ad} \right) \right) \nnorm{X} 
\\
  &= c_1 \Tr \left( \1_\V \right)
\\
  &= c_1\, n
\end{aligned}
\end{equation}
and, similarly, 
\begin{equation}
  \nnorm{X} = c_2\, n \, ,
\end{equation}
which proves the claimed implication. 
\end{proof}
%%% -------------------------- item -----------------------------------
\begin{proof}[Proof of \ref{item:reduced_eq} $\Rightarrow$ \ref{item:extremal}]
\renewcommand{\qedsymbol}{}
The crucial observation for this implication is that Assumption~\ref{item:reduced_eq}
alone assures that $Y^\sharp$ defined in Eq.~\eqref{eq:Y_sharp} with all off-diagonal blocks set to zero is a feasible point of Watrous' dual SDP, albeit not necessarily an optimal one. 
This claim is easily verified by direct computation. Inserting this dual feasible point into the SDP's objective function results in
\begin{equation*}
\begin{aligned}
% &\phantom{={.}}
\Tr \bigl[ D Y^\sharp \bigr]
% \\
&= 
\Tr\Bigl[
%\begin{pmatrix}
%1 & 0  &  \zerov_{\W \otimes \V}\ad &  \zerov_{\W \otimes \V}\ad\\
%0 & 1 & \zerov_{\W \otimes \V}^T &  \zerov_{\W \otimes \V}^T \\
%\zerov_{\W \otimes \V} & \zerov_{\W \otimes \V} & \zerom_{\W \otimes \V} & \zerom_{\W \otimes \V} \\
%\zerov_{\W \otimes \V} & \zerov_{\W \otimes \V} & \zerom_{\W \otimes \V} & \zerom_{\W \otimes \V} 
%\end{pmatrix}
\begin{pmatrix}
1 & 0  &  0 &  0\\
0 & 1 & 0 &  0 \\
0 & 0 & 0 & 0 \\
0 & 0 & 0 & 0 
\end{pmatrix}
 % \\
 % & \qquad 
  \cdot \frac{1}{2}
%  \begin{pmatrix}
%  \nnorm{X} & 0  &  \zerov_{\W \otimes \V}\ad &  \zerov_{\W \otimes \V}\ad\\
%  0 & \nnorm{X} & \zerov_{\W \otimes \V}^T &  \zerov_{\W \otimes \V}^T \\
%  \zerov_{\W \otimes \V} & \zerov_{\W \otimes \V} & n \sqrt{X X^\dagger} & \zerom_{\W \otimes \V} \\
%  \zerov_{\W \otimes \V} & \zerov_{\W \otimes \V} & \zerom_{\W \otimes \V} & n \sqrt{X^\dagger X} 
%  \end{pmatrix}
  \begin{pmatrix}
  \nnorm{X} & 0  &  0 &  0\\
  0 & \nnorm{X} & 0 & 0 \\
  0 & 0 & n \sqrt{X X^\dagger} & 0 \\
  0 & 0 & 0 & n \sqrt{X^\dagger X} 
  \end{pmatrix}
\Bigr]
\\
&= \frac{1}{2} \bigl( \nnorm{X} + \nnorm{X} \bigr)
= \nnorm{X}.
\end{aligned}
\end{equation*}
Since every dual SDP corresponds to a constrained minimization, evaluating the dual objective function at any feasible point results in an upper bound on the optimal value. 
In our case, we obtain the upper bound $\jnorm{X} \leq \nnorm{X}$, which together with the converse bound from Proposition~\ref{prop:Gaussian_CS}, 
implies equality between the two.
\end{proof}
% \end{description}
  
\end{proof}

%%% ==================================================
%%% ================= Discussion =========================
%%% ==================================================
\section{Discussion and outlook}

We conclude by mentioning several observations and research directions that may merit
further attention.

\paragraph*{Recovery guarantees}
In this work, we have shown that for matrices saturating the norm inequality 
\eqref{eq:nnorm_leq_jnorm}, recovery guarantees for square norm regularization are inherited from certain recovery guarantees for nuclear norm regularization. 
We give a geometric argument which makes it plausible that even a better performance can be expected and we identify it numerically. 
A promising route of future research is to develop methods allowing to prove recovery guarantees for the square norm directly. The tensorial character of the square norm is a challenge in such an endeavour that needs to be overcome and might also lead to new insights in other tensorial reconstruction problems. 

\paragraph*{Measurement errors}
In our analysis we considered reconstructed matrices $X^\mysquare_\eta$ and $X^\ast_\eta$ from Eqs.~\eqref{eq:Mdiamond} and \eqref{eq:Mast} that are required to be $\eta$-close to the ideal operator $X_0$. Such a reconstruction stably tolerates additive errors $\ev$ as in Eq.~\eqref{eq:measurement} as long as they obey $\fnorm{\ev} \leq \eta$. For operators $X_0$ satisfying the extremality \eqref{eq:equality} we prove that recovery guarantees for $X^\ast_\eta$ are inherited by $X^\mysquare_\eta$. 
A similar situation is true for the reconstruction of maps $M_0$ by means of diamond norm minimization.  For the idealized setting of noiseless measurements ($\ev =0$), we demonstrate numerically that often $\fnorm{M^\diamond_\eta - M_0}$ vanishes while $\fnorm{M^\ast_\eta - M_0}$ is large. 
A numerical analysis for the noisy case $\ev>0$ yields similar results as for $\ev=0$. For the noisy case the phase transition from having no recovery to almost always recovering the signal up to $\eta \gtrsim \fnorm{\ev}$ broadens equally for both diamond and nuclear norm regularization. 

\paragraph*{Partial derandomizations}
While initial theoretical results often rely on measurements that
follow a Gaussian distribution, later on significant effort has been
put into derandomizing the measurement process. On the one hand,
recovery guarantees for structured measurements were proven \cite{CanStroVor13}. On the other hand, also the distributions from which the measurements are drawn were partially derandomized \cite{GroKraKue15_partial,KueRauTer15,KabKueRau15} (see also Section~\ref{sec:application_low_rank}), relying on above mentioned $t$-designs. The later methods rely on an analysis of the measurement map's descent cone. Hence, such recovery guarantees for partially derandomized measurements are also inherited by our reconstruction via diamond norm minimization.
In a similar setting, a partial derandomization of the random unitaries used as part of the measurements for the retrieval of unitary basis changes (Section~\ref{sec:UV}) and for quantum process tomography (Section~\ref{sec:quantum}) seems very promising. Here, structural insights \cite{DanCleEme09,GroAudEis07,Zhu15,Web15} on unitary designs could be used in future work. 

\paragraph*{Improvement from structured measurements}
We numerically performed the reconstruction of unitary basis changes in Section~\ref{sec:UV} for two different measurement settings: Gaussian measurements and certain structured measurements. For the nuclear norm, the reconstruction from Gaussian measurements performed slightly better than the one from structured measurements, just as expected. Perhaps surprisingly, we observed the converse for the diamond norm reconstructions. Here, the structure of the measurements seems to be favourable for the reconstruction process. This observation motivates the search for recovery guarantees for diamond norm reconstruction with structured measurements. Such structured measurements are also crucial for the quantum process tomography in Section~\ref{sec:quantum} and blind matrix convolution in Section~\ref{sec:deconv}. 

\paragraph*{CPT as a constraint in the quantum channel reconstructions}
A map $M \in \M(\V,\W)$ is a quantum channel if and only if
\begin{equation}
  M\ad (\1_\W ) = \1_\V  
  \quad \text{and} \quad
  J(M) \succeq 0 \, .
\end{equation}
When aiming at reconstructing quantum channels, 
these additional constraints can, in principle, be included in the SDPs \eqref{eq:Mdiamond} and \eqref{eq:Mast} for the diamond norm and nuclear norm reconstructions. Doing so leads to a significant overhead in the numerical reconstruction process. Numerically, one can observe that the recovery success of the diamond norm reconstruction \eqref{eq:Mdiamond} is almost unchanged, while the nuclear norm reconstruction \eqref{eq:Mast} performs significantly better. In fact, it seems to perform roughly as well as the diamond norm reconstruction when these constraints are included in the SDP \eqref{eq:Mast}.
In this sense, the CPT structure can be used in the nuclear norm reconstruction at the expense of a longer computation time to reduce the number of measurements, while in the diamond norm reconstruction the CPT structure is already inbuilt. The run-time of the diamond norm reconstruction and the nuclear norm reconstruction are practically the same for a given number of measurements and scale polynomially with the number of constraints. Therefore, the diamond norm reconstruction can help to render larger quantum systems accessible to quantum process tomography.

\paragraph*{The robust rank null space property (NSP)}
This property is a certain norm inequality giving rise to yet a stronger version of uniform recovery guarantees \cite{KabKueRau15}. The bound \eqref{eq:nnorm_leq_jnorm} implies that if the NSP is fulfilled for the nuclear norm then it is also fulfilled with the diamond norm. 
Here, it is certainly a promising research endeavour to also search for recovery guarantees based on an NSP. 
We expect that besides that stability toward measurement noise such recovery guarantees can also feature robustness against violations of the model assumptions of low rank and, possibly, the saturation of the norm bound \eqref{eq:nnorm_leq_jnorm}.

\paragraph*{The noise parameter \texorpdfstring{$\eta$}{eta} in the reconstruction and other conceivable optimization methods}
Our reconstruction schemes are versions of the one in Eq.~\eqref{eq:general_reconstruction}. Here, an upper bound $\eta$ on the noise level needs to be given as an input to the reconstruction procedure. 
In applications one can, however, not always expect to have a good upper bound. 
Here, one can potentially resolve to other optimization methods. Instead of the optimization \eqref{eq:general_reconstruction}, the two following types of reconstructions are commonly used in compressed sensing. 

The first one is given by
% LASSO type
\begin{equation}
   x^{f} = \argmin_x \{ \fnorm{\A(x) - y} \ : f(x)\leq \tau\} \, ,
\end{equation} 
where $\tau>0$ is some parameter, which needs to be chosen. Denoting the original signal by $x_0$, one knows $f(x_0)$ in many applications such as in those presented in this work. Hence, in this case, one can choose $\tau = f(x_0)$. 

The second common optimization method is given by
% denoising
\begin{equation}
   x^{f,\lambda} = \argmin_x \{ \lambda\,f(x)+ \fnorm{\A(x) - y}\}
\end{equation}
for some fixed value $\lambda>0$. 

Solutions of the three optimization methods can be related to each other, which is made precise by \cite[Proposition~3.2]{FouRau13}. This proposition is formulated for a more specialized situation but it is clear that it also holds in greater generality. 

We leave a detailed comparison of different optimization methods with the square or diamond norm for future work. 

\section{Acknowledgements}
We thank C.\ Riofr\'io, I.\ Roth, and D.\ Suess for
insightful discussions and S.\ Kimmel and Y.K.\ Liu 
for advanced access to Ref.\ \cite{KimLiu15}. 
The research of MK and JE has been supported by the EU
(AQuS, RAQUEL), 
the BMBF (Q.com), the DFG projects EI 519/9-1 (SPP1798 CoSIP) and CRC183, the ERC (TAQ) and the Templeton Foundation. 
The work of DG and RK is supported by the Excellence Initiative of the German Federal and State Governments (Grants ZUK 43 \& 81), the ARO under contract W911NF-14-1-0098 (Quantum Characterization,
Verification, and Validation), the DFG (SPP1798 CoSIP), and the State Graduate Funding Program of Baden-W\"urttemberg.

%%% ================================================
%%% =============== APPENDIX ==========================
%%% ================================================
\appendix
\section{Appendix}
In this appendix we provide known material to make this work more self contained. 
We provide a brief introduction to tensor products and a basis independent definition of the Choi-Jamio{\l}kowski isomorphism in the first section. 
We devote the second section to low-rank matrix recovery. We show how the statements presented in Section~\ref{sec:application_low_rank} can be derived using geometric proof techniques. Unlike the first part of the appendix, this section does include technical novelties. 

\subsection{Basic concepts of multilinear algebra and the Choi-Jamio{\l}kowski isomorphism}\label{sec:appendixChoi}
The core objects of this work are tensors of order four and naturally fall into the realm of multilinear algebra. Here we give a brief introduction on core concepts of multilinear algebra that can be found in any textbook on that topic. Our presentation here is influenced by \cite{Wat11}.
Let $\V_1,\ldots,\V_k$ be (finite dimensional, complex) vector spaces with associated dual spaces $\V_1^\ast,\ldots,\V_k^\ast$. 
A function 
\begin{equation}
f: \V_1\times \cdots \times \V_k \to \CC
\end{equation}
is \emph{multilinear}, if it is linear in each $\V_i$. The space of such functions constitutes the \emph{tensor product} of $\V_1^\ast,\ldots,\V_k^\ast$ and we denote it by
$\V_1^\ast \otimes \cdots \otimes \V_k^\ast$. 
By reflexivity $\V \cong \V^{\ast \ast}$, the tensor product $\V_1 \otimes \cdots \otimes \V_k$ 
is the space of all multilinear functions
\begin{equation}
f: \V_1^\ast \times \cdots \times \V_k^\ast \to \CC.
\label{eq:tensor_product}
\end{equation}
Its elementary elements $z_1 \otimes \cdots \otimes z_k$ are the \emph{tensor product} of vectors $x_1 \in \V_1,\ldots, x_k \in \V_k$ which alternatively can be constructed by means of the Kronecker product -- however, such an explicit construction requires explicit choices of bases in $\V_1,\ldots,\V_k$.

With such a notation, the space of linear maps $\L(\V \to \W)$ (matrices) corresponds to the tensor product $\W \otimes \V^\ast$ which is spanned by rank-one operators $\left\{ y \otimes x^\ast: x \in \V, y \in \W \right\}$.
With this identification, it is straightforward to define the tensor product of
$\L \left( \W_1 \to \W_2 \right)$ and $\L \left( \V_1 \to \V_2 \right)$ to be
\begin{equation}
\begin{aligned}
% &\phantom{\cong{}}
\L \left( \W_1 \to \W_2 \right) \otimes \L \left( \V_1 \to \V_2 \right) 
% \\
&\cong 
\left( \W_2 \otimes \W_1^\ast \right) \otimes \left(  \V_2 \otimes \V_1^\ast \right)
\\
&\cong \L \left( \V_1 \otimes \W_1 \to \V_2 \otimes \W_2 \right) \, . 
\end{aligned}
\end{equation}
Analogously to before, the elementary $Y \otimes X$ of this space are the \emph{tensor product} 
of maps $Y \in \L \left( \W_1 \to \W_2 \right)$ and $X \in \L \left( \V_1 \to \V_2 \right)$. 
Restricting to tensor products of endomorphisms, i.e.\ $\W_2 \cong \W_1$ and $\V_2 \cong \V_1$, the \emph{partial trace} (over the first tensor factor) for elementary elements is defined to be
\begin{equation}
\begin{aligned}
\Tr_\W: \L(\W) \otimes \L (\V) & \to \L \left( \W \right) \\
Y \otimes X  & \mapsto  \Tr(X) \, Y
\end{aligned} 
\end{equation}
and extended linearly to $\L(\W) \otimes \L(\V)$. 
Note that with the identification $\L (\W) \cong \W \otimes \W^\ast$, $\Tr_\W$ corresponds to the natural contraction between $\W$ and $\W^\ast$. 
This is illustrated in Figure~\ref{fig:tensor_diagram}.

Similarly to $\L \left( \V_1 \to \V_2 \right)$, the maps $\L(\L(\V_1 \to \V_2) \to \L(\W_1 \to \W_2))$ introduced in Section~\ref{sec:notation_maps} can be viewed as elements of the tensor product space
\begin{equation}
\left( \W_2 \otimes \W_1^\ast \right) \otimes \left( \V_2 \otimes \V_1^\ast \right)^\ast
\cong \W_2 \otimes \W_1^\ast \otimes \V_2^\ast \otimes \V_1 \, ,
\label{eq:four_linear_tensors}
\end{equation}
which can be seen as a four-linear vector space. There are several equivalent ways to interpret its elements. For the given applications of our work, we have made heavy use of the \emph{Choi-Jamio{\l}kowski isomorphism}
which acts on four-linear tensors by permuting tensor factors:
\begin{equation}
\begin{aligned}
\tilde J : \V_1 \otimes \V_2\otimes \V_3 \otimes \V_4 
&\to 
\V_1 \otimes \V_3\otimes \V_2 \otimes \V_4 \, ,
\\
  v_1\otimes v_2 \otimes v_3 \otimes v_4 
 &\mapsto
 v_1\otimes v_3 \otimes v_2 \otimes v_4 \, . 
\end{aligned}
\end{equation}
Applied to the four-linear space of maps \eqref{eq:four_linear_tensors}
we obtain
\begin{equation}
\begin{aligned}
  % &\phantom{\cong{}}
 \L(\L(\V_1 \to \V_2) \to \L(\W_1 \to \W_2)) 
 % \\
 &\cong
 \L(\V_2 \otimes \V_1^\ast \to \W_2 \otimes \W_1^\ast )
 \\
 &\cong
 \W_2\otimes \W_1^\ast \otimes \V_2^\ast\otimes \V_1 \, ,
\end{aligned} 
\end{equation}
and 
\begin{equation}
 \L(\W_1\otimes \V_1^\ast \to \W_2\otimes \V_2^\ast) 
  \cong 
 \W_2\otimes \V_2^\ast \otimes \W_1^\ast \otimes \V_1
\end{equation}
which are basis independent. 
Consequently, the Choi-Jamio{\l}kowski isomorphism is linear bijection from
maps to operators,
\begin{equation}
 \begin{aligned}
\tilde J: \
\L(\L(&\V_1 \to \V_2) \to \L(\W_1 \to \W_2)) 
% \\ 
\ \to \
\L(\W_1 \otimes \V_1^\ast \to \W_2\otimes \V_2^\ast) \, .
 \end{aligned}
\end{equation}
Its explicit definitions \eqref{eq:choi} and \eqref{eq:Jdef} in the main text
are just basis-dependent realization of this more general identification, where 
$\V = \V_1 = \V_2$ and $\W = \W_1 = \W_2$. 
Moreover, this realization requires the identification $\V \cong \V^*$ via an inner product on $\V$. This identification also induces a complex conjugation and the corresponding transposition on $\L(\V)$.
We illustrated this fact pictorially in  Figure~\ref{fig:tensor_diagram}
by resorting to \emph{tensor network} \cite{WerJasSil16} or \emph{wiring diagrams} \cite{Lan12}. 

\subsection{Uniform recovery guarantees and partial derandomizations} \label{sub:uniform}
Our main geometric insight -- Corollary \ref{cor:subset} -- asserts that any square norm descent cone is always contained in the corresponding one of the nuclear norm, provided that the operators in question obey
$\jnorm{X} = \nnorm{X}$.
When applying this idea to low-rank matrix recovery, we started with mentioning 
Proposition~\ref{prop:Gaussian_CS}. 
This is a non-uniform recovery guarantee that is stable towards additive noise. 
However, with some additional work, Corollary \ref{cor:subset} allows for stronger conclusions.
Some of them are summarized in Proposition~\ref{prop:fourth_moments} and Proposition~\ref{prop:rank_one}, respectively. 
Here, we outline how these results are obtained. 
In Section \ref{sub:convex_recovery} we introduced widely used geometric proof techniques for low-rank matrix recovery mainly following Ref.\ \cite{Tro14}. 
These aim at recovery of a fixed object $X_0$ of interest and thus it suffices to focus on precisely one descent cone, namely $\DC(X_0,\nnorm{\argdot})$, or $\DC(X_0,\dnorm{\argdot})$, respectively.
By taking a closer look at the actual proof techniques -- most notably Mendelson's small ball method \cite{Men14}, or Tropp's bowling scheme \cite{Tro14} -- one can see that such a restriction to a single object of interest is not necessary. 
Up to our knowledge, this was first pointed out in Ref.\ \cite{KueRauTer15} and 
at the heart of this observation is the following technical statement.

\begin{lemma} \label{lem:effective_low_rank}
Fix $1 \leq r \leq n$ and let $K_r = \bigcup_{\rank(X)=r} \DC(\nnorm{\argdot},X ) \subset \L (\V)$ be the union of all descent cones 
anchored in nonzero matrices $0 \neq X \in \L (\V )$ of rank at most $r$.
Then, every element $Y \in K_r$ obeys 
\begin{equation}
\nnorm{Y} \leq (1+\sqrt{2}) \sqrt{r} \fnorm{Y}\, .
\end{equation}
\end{lemma}

For Hermitian matrices, a slightly stronger statement of this type was presented in \cite[Lemma 10]{KueRauTer15}. Here, we provide a different proof that does not require hermiticity and exploits a variant of pinching.

\begin{lemma}[Pinching inequality] \label{thm:pinching} 
Let $P,Q \in \L (\V)$ be orthogonal projectors with complements $P^\perp = \1_{\V}-P$ and 
$Q^\perp = \1_{\V}-Q$.
Also, let $\norm{\argdot}_p$ be any Schatten-$p$ norm. Then, every $Z \in \L (\V)$ obeys
\begin{equation}
\| P Z Q \|_p^p + \| P^\perp Z Q^\perp\|_p^p
\leq \| Z \|_p^p \, . 
\label{eq:pinching}
\end{equation}

\end{lemma}

\begin{proof}
Note that for any $Z \in \L (V)$ it follows from the definition of the Schatten-$p$ norms that the left hand side of Eq.~\eqref{eq:pinching} coincides with
$ \| P Z Q + P^\perp Z Q^\perp \|_p^p$. Using this identity and the decomposition
\begin{equation}
P Z Q + P^\perp Z Q^\perp 
= \frac{1}{2} Z + \frac{1}{2} \bigl(P-P^\perp \bigr) Z \bigl( Q - Q^\perp \bigr)
\end{equation}
allows us to conclude
\begin{equation}
\begin{aligned}
% &\phantom{={.}} 
\| P Z Q + P^\perp Z Q^\perp \|_p^p 
% \\
 &= \normb{\frac{1}{2} Z + \frac{1}{2} \bigl(P-P^\perp \bigr) Z \bigl( Q - Q^\perp \bigr)}_p^p 
\\
 &\leq  \frac{1}{2} \| Z \|_p^p + \frac{1}{2} \normb{\bigl(P-P^\perp\bigr)Z \bigl(Q- Q^\perp\bigr)}_p^p 
\\
 &= \frac{1}{2} \| Z \|_p^p +  \frac{1}{2} \| Z \|_p^p
 = \| Z \|_p^p \, ,
\end{aligned}
\end{equation}
where we have exploited unitary invariance of Schatten-$p$ norms and the fact that both $P-P^\perp$ and $Q - Q^\perp$ are unitary matrices.
\end{proof}

\begin{proof}[Proof of Lemma~\ref{lem:effective_low_rank}]
It suffices to prove this statement for any fixed descent cone $ \DC(\nnorm{\argdot},X )$, where $X \in \L (\V)$ has rank at most $r$.
Let $\C\coloneqq \ran(X)$ and $\R \coloneqq \ran(X\ad)$ be the column and row ranges of $X$ 
(these need not coincide, since $X$ need not necessarily be Hermitian) 
and let $P_\C,P_\R \in \L (\V)$ be orthogonal projections onto these subspaces. 
Note that if $X$ has a singular value decomposition $X = U \Sigma V\ad$, then $P_\C = U \Sigma^0 U\ad$ and $P_\R = V \Sigma^0 V\ad$, where $\Sigma^0$ is defined component-wise by $\Sigma^0_{i,j}\coloneqq 1$ if $\Sigma_{i,j} \neq 0$ and $\Sigma^0_{i,j}\coloneqq 0$ otherwise. Introducing orthogonal complements $P_\C^\perp = \1_{\V(\L)} - P_\C$ and $P_\R^\perp = \1_{\L (\V)}-P_\R$ allows us to define 
\begin{equation}
\mathcal{P}_T^\perp: \L (\V) \to \L (\V), \quad
Z \mapsto P_\C^\perp Z P_\R^\perp \, .
\end{equation}
This is an orthogonal projection with respect to the Frobenius inner product \eqref{eq:inner_prod} and 
obeys $\mathcal{P}_T^\perp(X) = 0$ by construction.
Its complement amounts to 
\begin{equation}
\mathcal{P}_T(Z) = Z - P_\C^\perp Z P_\R^\perp
= P_\C Z + Z P_\R - P_\C Z P_\R 
\end{equation}
which obeys $\mathcal{P}_T(X) = X$. Note that this is a straightforward generalization of the $T$-space introduced in \cite[Equation (2)]{Gro11} to non-Hermitian matrices.
Analogously to there, a decomposition $Z = Z_T + Z_T^\perp := \mathcal{P}_T(Z) + \mathcal{P}_T^\perp(Z)$ is valid for every $Z \in \L (\V)$ and every  $Z_T := \mathcal{P}_T(Z)$ has rank at most $2r$ by construction.

Now choose $Y \in  \DC(\nnorm{\argdot},X )$ and note that by definition
$\nnorm{X} \geq \nnorm{X+\tau Y}$ must be valid for some $\tau >0$. Combining this with Lemma~\ref{thm:pinching} (Pinching) assures
\begin{equation}\label{eq:app_aux1}
\begin{aligned}
% &\phantom{={}}
\nnorm{X}
% \\
 &\geq  \nnorm{X + \tau Y}\\
&\geq \nnorm{P_\C (X+\tau Y) P_\R} + \nnormb{P_\C^\perp (X+\tau Y) P_\R^\perp} 
\\
&=  
\nnorm{X+\tau P_\C Y P_\R} + \nnormb{ \mathcal{P}_T^\perp (X + \tau Y)}
\\
&= 
\nnorm{X+\tau P_\C Y P_\R} + \tau \nnormb{Y_T^\perp} \, , 
\end{aligned}
\end{equation}
where we have employed $P_\C X P_\R = X$ and $\mathcal{P}_T^\perp(X) = 0$.
Also, note that H\"older's inequality assures $|\Tr \left( U Z \right)|\leq \nnorm{Z}$ for any $Z \in \L (\V)$ and unitary $U$. 
Employing this for $U = S_X$, where the sign matrix $S_X$ of $X$ was defined in Definition~\ref{def:sgn}, 
reveals that 
\begin{equation}\label{eq:app_aux2}
\begin{aligned}
\nnorm{X+\tau P_\C Y P_\R}
& \geq 
\Tr \left( S_X X \right) + \tau\, \left|\Tr \left( S_X P_\C Y P_\R \right)\right| 
\\
 &\geq  
\nnorm{X} - \tau \snorm{S_X} \nnorm{P_\C Y P_\R} 
\\
&\geq 
\nnorm{X}-\tau \sqrt{r} \fnorm{P_\C Y P_\R}
\\
&\geq 
\nnorm{X} - \tau \sqrt{r} \fnorm{Y},
\end{aligned}
\end{equation}
where we have in addition used that $P_\C Y P_\R$ has rank at most $r$ and Frobenius norm smaller than or equal to $\fnorm{Y}$. Combining the bounds~\eqref{eq:app_aux1} and \eqref{eq:app_aux2} implies 
\begin{equation}
\nnorm{X} \geq \nnorm{X} + \tau \left( \nnormb{Y_T^\perp} - \sqrt{r} \fnorm{Y} \right).
\end{equation}
Since $\tau >0$, this bound implies
$\nnorm{Y_T^\perp} \leq \sqrt{r} \fnorm{Y}$.
Finally, this relation allows us to infer the result,
\begin{equation}
\begin{aligned}
\nnorm{Y} &= \nnormb{Y_T+Y_T^\perp}
\\
&\leq 
\nnorm{Y_T} + \nnormb{Y_T^\perp}
\\
&\leq 
\sqrt{2r} \fnorm{Y_T} + \sqrt{r} \fnorm{Y}
\\
&= (1+\sqrt{2})\sqrt{r} \fnorm{Y} \, ,
\end{aligned}
\end{equation}
where we also exploited the fact that $Y_T$ has rank at most $2r$.

\end{proof}

Lemma~\ref{lem:effective_low_rank}
asserts that any matrix that lies in the nuclear norm's descent cone of any low-rank matrix, is ``effectively'' a low-rank matrix as well. 
This structural property together with Mendelson's small ball method \cite{Men14,KolMen14} is enough to bound the minimal conic singular value of a measurement map $\A$ with respect to the union of all possible descent cones.
Here we provide a particular realization of Mendelson's small ball method that is directly applicable to low-rank matrix recovery (see e.g.\ Ref.\ \cite[Section 4]{KueRauTer15}).

\begin{theorem}[A variant of Mendelson's small ball method] \label{thm:mendelson}
Let $\mc L\subset \L(\V)$  be real subspace of linear maps and 
let $\A: \mc L \to \RR^m$ be a measurement map $\A(X) = \sum_{i=1}^m \Tr \left( A_i X \right) e_i$, where each $A_i$ is an independent copy of a random matrix $A \in \L (\V)$ and $e_1,\ldots,e_m$ denotes the standard basis in $\RR^m$. Also, let $E_r = \left\{ Y \in K_r: \; \fnorm{Y}=1 \right\}$,
where $K_r$ was defined in Lemma~\ref{lem:effective_low_rank}.
Then for  any $\xi,t >0$, the bound
\begin{equation}
\lambda_{\mathrm{min}} \left( \A, K_r \right) \geq \xi \sqrt{m} Q_{2 \xi} \left( E_r ; A \right)
- 2 W_m \left( E_r, \A \right) - \xi t
\end{equation}
holds with probability at least $1 - \mathrm{e}^{-2t^2}$. Here,
\begin{equation}
\begin{aligned}
Q_\xi \left( E_r, A \right) 
&= 
\inf_{Y \in E_r} \mathrm{Pr} \left[ \bigl| \Tr( A\ad Y ) \bigr| \geq \xi \right] ,
\\
W_m \left( E_r, \A \right) 
&= 
\EE \left[ \sup_{Y \in E_r} \Tr( H\ad Y) \right] , 
\end{aligned}
\end{equation}
where
\begin{equation}
H = \frac{1}{\sqrt{m}} \sum_{j=1}^m \epsilon_j A_j
\end{equation}
and $\epsilon_1,\ldots,\epsilon_m$ being a Rademacher sequence\footnote{
A Rademacher sequence is a sequence of independent random variables that take the values $\pm 1$ with equal probability.}.
\end{theorem}

Important examples for the space of considered operators are $\mc L = \Herm(\V)$ and real matrices. 

Thanks to Lemma~\ref{lem:effective_low_rank} and H\"older's inequality we can bound $W_m \left( E_r, \A \right)$ in Theorem~\ref{thm:mendelson} by
\begin{equation}\label{eq:Wm}
\begin{aligned}
W_m \left( E_r, \A \right) 
&= \EE \left[ \sup_{Y \in E_r} \Tr \left( H\ad Y \right) \right]
\\
&\leq \EE \left[ \sup_{Y \in E_r} \nnorm{Y} \snormb{H\ad} \right] 
\\
&\leq \EE \left[ \sup_{Y \in E_r} (1+\sqrt{2}) \sqrt{r} \fnorm{Y} \snorm{H} \right] 
\\
&= (1+\sqrt{2}) \sqrt{r}\, \EE \left[ \snorm{H} \right] \, ,
\end{aligned}
\end{equation}
which is much easier to handle. 
This simplification together with Mendelson's small ball method -- Theorem~\ref{thm:mendelson} -- and the geometric error bound for convex recovery -- Proposition~\ref{prop:general_reconstruction} -- provide a convenient sufficient means to assure that a given measurement process $\A$ allows for 
uniform and stable low-rank matrix recovery via nuclear norm minimization:

\begin{proposition} [Sufficient criteria for uniform recovery]\label{prop:sufficient_criteria}
Let $\A: \L (\V) \to \CC^m$ be a measurement map as defined in Theorem~\ref{thm:mendelson} and fix $1 \leq r \leq n$. Suppose that this measurement map obeys
$
Q_{2 \xi}(E_r; \A) \geq C_1 
$ for some $\xi >0$ and also $ \EE \left[ \snorm{H} \right] \leq C_2 \sqrt{m/r}$, where $C_1$ and $C_2$ are positive constants
obeying $\xi C_1 > 2(1+\sqrt{2})C_2$.

Then, with probability at least $1  - \mathrm{e}^{-C^\ast_4 m}$, this measurement map is capable of stably reconstructing any matrix $X_0$ of rank at most $r$ 
from noisy measurements of the form $y = \A (X_0) + \epsilon$ obeying $\fnorm{\epsilon} \leq \eta$ 
by means of nuclear norm minimization. Concretely, the solution $X_\eta^\ast$ of the optimization \eqref{eq:Mast} obeys
\begin{equation}
 \fnorm{X^\ast_\eta - X_0} \leq \frac{\eta}{C^\ast_3\sqrt{m}}.
\end{equation}
Here $C^\ast_3,C^\ast_4 >0$ denote sufficiently small absolute constants.
\end{proposition}

Note that unlike Proposition~\ref{prop:Gaussian_CS}, such a recovery statement is \emph{uniform}, in the sense that with high probability a single measurement map 
suffices to recover any low-rank matrix. 
However, it still relies on the geometric proof technique of bounding the widths of nuclear norm descent cones. This is because the set $K_r$ is just the union over all possible nuclear norm descent cones anchored at matrices of rank at most $r$. As a result, Observation~\ref{obs:smaller_DC} (``the smaller the descent cone, the better the recovery'') is also valid in this setting and Corollary~\ref{cor:subset}
allows us to draw the following conclusion.

\begin{corollary} [Uniform recovery from square norm regularization]
\label{cor:sufficient_diamond}
The assertions of Proposition~\ref{prop:sufficient_criteria} remain true for recovery via square norm regularization \eqref{eq:Mdiamond},
for the case of uniform recovery of rank-$r$ maps 
$X_0 \in \L \left( \V \otimes \W \right)$ satisfying $\jnorm{X_0} = \nnorm{X_0}$. 
Moreover, the corresponding constants obey $C_3^\mysquare \geq C_3^\ast$ and $C_4^\mysquare \geq C_3^\ast$, meaning that the recovery statement cannot be worse.
\end{corollary}

\begin{proof}[Proof of Proposition~\ref{prop:sufficient_criteria}]
Theorem~\ref{thm:mendelson} together with Eq.~\eqref{eq:Wm} and the assumptions on $\A$ assure for any $t > 0$
\begin{equation}\label{eq:lambda_min_bound}
\begin{aligned}
 % & \phantom{\geq{}} 
\lambda_{\mathrm{min}} \left( \A, K_r \right) 
 % \\
&\geq \xi \sqrt{m} Q_{2 \xi}(E_r; \A ) - 2 W_m (E_r,\A) - \xi t \\
&\geq \xi \sqrt{m} Q_{2 \xi} (E_r; \A) - 2 (1+\sqrt{2}) \sqrt{r} \EE \left[ \snorm{H} \right] - \xi t \\
& \geq \xi C_1 \sqrt{m} - 2(1+\sqrt{2}) C_2 \sqrt{m} - \xi t 
\end{aligned}
\end{equation}
with probability at least $1  - \mathrm{e}^{-2t^2}$.
Introducing $C_3 = (\xi C_1 - 2(1+\sqrt{2}) C_2)/2$ -- which is strictly positive by assumption -- 
and setting $t = C_3 \sqrt{m}/\xi$ then implies
\begin{equation}
\lambda_{\mathrm{min}} \left( \A, K_r \right)  \geq C_3 \sqrt{m}
\end{equation}
with probability at least $1  - \mathrm{e}^{- C_4 m}$, where $C_4 = C_3^2 / \xi^2 >0$. 
With such an estimate at hand, the claim follows from applying Proposition~\ref{prop:general_reconstruction}.
\end{proof}

We conclude this section with presenting a selection of measurement ensembles that meet the criteria of Proposition~\ref{prop:sufficient_criteria} and as a consequence also the ones of Corollary~\ref{cor:sufficient_diamond}. 
We start with measurement ensembles that allow for recovering real-valued matrices 
$X \in \L (\V)$.

\begin{corollary} \label{cor:derandomizations1}
Suppose that $\V$ is a real-valued vector spaces and let $\A: \L (\V) \to \RR^m$ be the measurement map $\A (X) = \sum_{i=1}^m \Tr \left( A_i X \right) e_i$, where each $A_i$ is a random matrix with independent entries obeying
\begin{equation}
\EE \left[ a_{i,j} \right] = 0,\quad 
\EE \left[ a_{i,j}^2 \right] = 1, \quad
\EE \left[ a_{i,j}^4 \right] \leq F,
\end{equation}
where $F$ is a constant. Then a sampling rate of $m \geq C r n$ suffices to
meet the requirements of Proposition~\ref{prop:sufficient_criteria}.
\end{corollary}

The result quoted in Corollary~\ref{cor:derandomizations1} was not established as a subroutine of a geometric proof technique for nuclear norm recovery, but consists of auxiliary statements that help to establish the Frobenius stable null space property
\cite[Definition 10]{KabKueRau15} -- a powerful alternative to geometric proof techniques relying on Proposition~\ref{prop:general_reconstruction}.
However, if embedded properly into the framework of geometric recovery proof techniques, the auxiliary statements in Ref.\ \cite{KabKueRau15} -- see also Ref.\ \cite{KabRauTer15a,KabRauTer15b} -- 
can still be used to establish recovery guarantees that rely on bounding the widths of descent cones. 
For our purposes, such a geometric proof environment is crucial, and this entire section is devoted to develop it. However, we point out that introducing and analyzing the square norm analogue of the Frobenius stable null space property -- which is geared towards nuclear norm minimization -- does constitute an intriguing follow-up problem. We leave this to future work.

\begin{proof}[Proof of Corollary~\ref{cor:derandomizations1}]
For a proof of this statement, we utilize auxiliary statements from Ref.\ \cite{KabRauTer15a}. 
Lemma 11 in loc.\ cit.\ asserts that such random matrices with bounded fourth moments obey $Q_{1/\sqrt{2}} \geq 1/4\max \left\{3, F \right\}$, where $F$ is the fourth-moment bound. 
Also, Ref.\ \cite[Lemma 12]{KabRauTer15a} assures $\EE \left[ \snorm{H} \right] \leq C_F \sqrt{n}$,
where $C_F$ is a constant that only depends on $F$. This, in particular, assures that
\begin{equation}
\EE \left[ \snorm{H} \right] \leq C_F \sqrt{n} \leq \frac{C_F}{\sqrt{C}} \sqrt{\frac{m}{r}}
\end{equation}
and we can set $\xi = 2^{-3/2}$, $C_2 = C_F / \sqrt{C}$ and $C_1 = 1/4 \max \left\{3,F \right\}$. 
Choosing the constant $C$ in the sampling rate large enough assures that these constants obey $\xi C_1 > 2(1+\sqrt{2}) C_2$ for $\xi = 2^{-3/2}$
and all the requirements of Proposition~\ref{prop:sufficient_criteria} are met. The claim then follows from applying this statement.
\end{proof}

We conclude this section with embedding the main results of Ref.\ \cite{KueRauTer15} into this framework.
In fact, the entire apparatus presented in this section is a condensed version of the proofs in loc.\ cit. However, the reader's convenience, we include the corresponding statement here as well.

\begin{corollary} \label{cor:derandomizations2}
Consider measurement maps $\A: \Herm(\V) \to \RR^m$ of the form $\A (X) = \sum_{i=1}^m \Tr \left( A_i X \right) e_i$.
Then the following measurement ensembles meet the requirements of Proposition~\ref{prop:sufficient_criteria}, if restricted to the recovery of Hermitian matrices:
\begin{enumerate}
\item $m \geq C_G r n$ and each $A_i = a_i a_i\ad$ corresponds to the outer product of a complex standard Gaussian vector $a_i \in \V$ with itself,
\item $m \geq C_{4D} r n \log (2n)$ and each $A_i = a_i a_i\ad$ is the outer product of a randomly selected element $a_i$ of a complex projective $4$-design.
\end{enumerate}
Once more, $C_G$ and $C_{4D}$ denote sufficiently large constants.
\end{corollary}

\begin{proof}

Let us start with the Gaussian case. In Ref.\ \cite[Section 4.1.]{KueRauTer15} 
the bounds $Q_{1/\sqrt{2}}~\geq~1/96$ and $\EE \left[ \snorm{H} \right] \leq c_1 \sqrt{n}$ are derived under the assumption $m \geq c_2 n$, where $c_1$ is sufficiently large. Thus, similarly to the proof of Corollary~\ref{cor:derandomizations1}, setting $\xi = 2^{-3/2}$ and choosing the constant $C_G$ in $m$ sufficiently large indeed meets the requirements of Proposition~\ref{prop:sufficient_criteria}.

For the $4$-design case, \cite[Proposition 12]{KueRauTer15} assures that the bound
$Q_{\xi} \left( E_r, \A \right) \geq  \left( 1 - \xi^2 \right)^2/24$ is valid for any $\xi \in [0,1]$.
Also, Ref.\ \cite[Proposition 13]{KueRauTer15} implies 
\begin{equation}
\EE \left[ \snorm{H} \right]
\leq 3.1049 \sqrt{ n \log (2n)}
\leq \frac{3.1049}{\sqrt{C_{4D}}} \sqrt{\frac{m}{r}},
\end{equation}
where we have inserted $m \geq C_{4D} rn \log (2n)$. 
Thus, choosing $\xi$ appropriately and the constant $C_{4D}$ in the sampling rate $m$ large enough again assures that the requirements of Proposition~\ref{prop:sufficient_criteria} are met.
\end{proof}

%%% ----------------------------------------------------------------------------
%%% ------------------------------ Bibliography --------------------------------
%%% ----------------------------------------------------------------------------
\bibliographystyle{apsrev4-1}
\bibliography{martin}

\end{document}